\documentclass[11pt]{article}
\usepackage{graphicx}
\usepackage{parskip}
\usepackage{fullpage}
\usepackage{tcolorbox}
\usepackage{tikz}
\usetikzlibrary{positioning}
\usetikzlibrary{backgrounds}
\usepackage{xcolor}
\usepackage{bm}
\usepackage{dsfont}

\usepackage[charter]{mathdesign}
\usepackage{mathtools}
\usepackage{amsthm}
\usepackage{doi}
\usepackage[inline]{enumitem}
\setlist{noitemsep,leftmargin=*}
\hypersetup{colorlinks,linkcolor={red!50!black},citecolor={blue!50!black},
  urlcolor={blue!80!black}}
\newcommand{\Xomit}[1]{}
\usepackage{CJK}
\newtheorem{theorem}{Theorem}
\newtheorem{lemma}{Lemma}[section]

\title{Improved Algorithms for Recognizing Perfect Graphs and
Finding Shortest Odd and Even Holes}

\author{Yung-Chung Chiu\footnote{Department of Computer Science and
Information Engineering, National Taiwan University.}
\and
Kai-Yuan Lai\footnote{Department of Computer Science and
Information Engineering, National Taiwan University.}
\and
Hsueh-I Lu\footnote{Corresponding author.  Department of Computer
Science and Information Engineering, National Taiwan University.
Email: hil@csie.ntu.edu.tw.  Research of this author is supported by
MOST grant 110--2221--E--002--075--MY3.}}

\begin{document}
\begin{CJK*}{UTF8}{bkai}
\maketitle

\begin{abstract}  
An \emph{induced} subgraph of an $n$-vertex graph $G$ is a graph that
can be obtained by deleting a set of vertices together with its
incident edges from $G$.  A \emph{hole} of $G$ is an induced cycle of
$G$ with length at least four. A hole is \emph{odd} (respectively,
\emph{even}) if its number of edges is odd (respectively,
even). Various classes of induced subgraphs are involved in the
deepest results of graph theory and graph algorithms. A prominent
example concerns the \emph{perfection} of $G$ that the chromatic
number of each induced subgraph $H$ of $G$ equals the clique number of
$H$. The seminal Strong Perfect Graph Theorem proved in 2006 by
Chudnovsky, Robertson, Seymour, and Thomas, conjectured by Berge in
1960, confirms that the perfection of $G$ can be determined by
detecting odd holes in $G$ and its complement. Based on the theorem,
Chudnovsky, Cornu{\'{e}}jols, Liu, Seymour, and Vu\v{s}kovi\'{c} show
in 2005 an $O(n^9)$-time algorithm for recognizing perfect graphs,
which can be implemented to run in $O(n^{6+\omega})$ time for the
exponent $\omega<2.373$ of square-matrix multiplication. We show the
following improved algorithms for detecting or finding induced
subgraphs in $G$.
\begin{enumerate}
\item 
The tractability of detecting odd holes in $G$ was open for decades
until the major breakthrough of Chudnovsky, Seymour, and Spirkl in
2020. Their $O(n^9)$-time algorithm is later implemented by Lai, Lu,
and Thorup to run in $O(n^8)$ time, leading to the best formerly known
algorithm for recognizing perfect graphs. Our first result is an
$O(n^7)$-time algorithm for detecting odd holes, immediately implying
a state-of-the-art $O(n^7)$-time algorithm for recognizing perfect
graphs. Finding an odd hole based on Chudnovsky et al.'s $O(n^9)$-time
(respectively, Lai et al.'s $O(n^8)$-time) algorithm for detecting odd
holes takes $O(n^{10})$ (respectively, $O(n^9)$) time. Nonetheless,
our algorithm finds an odd hole within the same $O(n^7)$ time bound.

\item 
Chudnovsky, Scott, and Seymour extend in 2021 the $O(n^9)$-time
algorithms for detecting odd holes (2020) and recognizing perfect
graphs (2005) into the first polynomial-time algorithm for obtaining a
shortest odd hole in $G$, which runs in $O(n^{14})$ time. Our second
result is an $O(n^{13})$-time algorithm for finding a shortest odd
hole in $G$.

\item 
Conforti, Cornu{\'e}jols, Kapoor, and Vu\v{s}kovi\'{c} show in 1997
the first polynomial-time algorithm for detecting even holes, running
in about $O(n^{40})$ time. It then takes a line of intensive efforts
in the literature to bring down the complexity to $O(n^{31})$,
$O(n^{19})$, $O(n^{11})$, and finally $O(n^9)$. On the other hand, the
tractability of finding a shortest even hole in $G$ has been open for
16 years until the very recent $O(n^{31})$-time algorithm of Cheong
and Lu in 2022. Our third result is two improved algorithms for
finding a shortest even hole in $G$ that run in $O(n^{25})$ and
$O(n^{23})$ time, respectively.
\end{enumerate}
\end{abstract}

\section{Introduction}
\label{section:section1}
Let $G$ be an $n$-vertex undirected and unweighted graph.  Let $V(G)$
consist of the vertices of $G$.  For any graph $H$, let $G[H]$ be the
subgraph of $G$ induced by $V(H)$.  A subgraph $H$ of $G$ is
\emph{induced} if $G[H]=H$.  That is, an induced subgraph of $G$ is a
graph that can be obtained from $G$ by deleting a set of vertices in
tandem with its incident edges.  To \emph{detect} an (induced) graph
$H$ in $G$ is to determine whether $H$ is isomorphic to an (induced)
subgraph of $G$.  To \emph{find} an (induced) graph $H$ in $G$ is to
report an (induced) subgraph of $G$ that is isomorphic to $H$, if
there is one.  Various classes of induced subgraphs are involved in
the deepest results of graph theory and graph algorithms.  One of the
most prominent examples concerns the \emph{perfection} of $G$ that the
chromatic number of each induced subgraph $H$ of $G$ equals the clique
number of $H$.  A graph is \emph{odd} (respectively, \emph{even}) if
it has an odd (respectively, even) number of edges.  A \emph{hole} of
$G$ is an induced cycle of $G$ having at least four edges.  The
seminal Strong Perfect Graph~Theorem of Chudnovsky, Robertson,
Seymour, and Thomas~\cite{ChudnovskyRST06,ChudnovskyS09}, conjectured
by Berge in 1960~\cite{Berge60,Berge61,Berge85}, confirms that the
perfection of a graph $G$ can be determined by detecting odd holes in
$G$ and its complement.  Based on the theorem, the first known
polynomial-time algorithms for recognizing perfect graphs
take~$O(n^{18})$~\cite{CornuejolsLV03} and
$O(n^9)$~\cite{ChudnovskyCLSV05} time.  The $O(n^9)$-time version can
be implemented to run in~$O(n^{6+\omega})$ time~\cite[\S6.2]{LaiLT20}
via efficient algorithms for the \emph{three-in-a-tree}
problem~\cite{ChudnovskyS10} that detects induced subtrees of $G$
spanning three prespecified vertices, where
$\omega<2.373$~\cite{AlmanW21,CoppersmithW90,LeGall14,VassilevskaWilliams12}
is the exponent of square-matrix multiplication.

Detecting induced subgraphs, even the most basic ones like paths,
trees, and cycles, is usually more challenging than detecting their
counterparts that need not be induced.  For instance, detecting paths
spanning three prespecified vertices is tractable~(via,
e.\,g.,~\cite{KawarabayashiKR12,RobertsonS95b}). However, the
\emph{three-in-a-path} problem that detects induced paths spanning
three prespecified vertices is NP-hard~(see,
e.\,g.,~\cite{HaasH06,LaiLT20}).  The \emph{two-in-a-path} problem
that detects induced paths spanning two prespecified vertices is
equivalent to determining whether the two vertices are connected.
Nonetheless, the corresponding two-in-an-odd-path and
two-in-an-even-path problems are
NP-hard~\cite{Bienstock91,Bienstock92}, whose state-of-the-art
algorithms on a planar graph take $O(n^7)$ time~\cite{KaminskiN12}.
Finding a non-shortest $uv$-path is easy.  A $k$-th shortest $uv$-path
can also be found in near linear time~\cite{Eppstein98}.
Nevertheless, the first polynomial-time algorithm for finding a
non-shortest $uv$-path takes $O(n^{18})$ time~\cite{BergerSS21-dm},
which is reduced to $\tilde{O}(n^{2\omega})$ time very
recently~\cite{ChiuL22}.

Detecting trees spanning a given set of vertices is easy via the
connected components, but detecting induced trees spanning a set of
prespecified vertices is NP-hard~\cite{GolovachPL12}.  The
three-in-a-tree problem is shown to be solvable first in $O(n^4)$
time~\cite{ChudnovskyS10} and then in $\tilde{O}(n^2)$
time~\cite{LaiLT20} via involved structural theorems and dynamic data
structures.  The tractability of the corresponding $k$-in-a-tree
problem for any fixed $k\geq 4$ is still unknown, although the problem
can be solved in $O(n^4)$ time on a graph of girth at
least~$k$~\cite{LiuT10}.

Cycle detection has a similar situation.  Detecting cycles of length
three, which have to be induced, is the classical triangle detection
problem that can be solved efficiently by matrix multiplications~(see,
e.\,g.,~\cite{WilliamsW18}).  It is tractable to detect cycles of
length at least four spanning two prespecified vertices (via,
e.\,g.,~\cite{KawarabayashiKR12,RobertsonS95b}), but the
\emph{two-in-a-cycle} problem that detects holes spanning two
prespecified vertices is NP-hard (and so are the corresponding
one-in-an-even-cycle and one-in-an-odd-cycle
problems)~\cite{Bienstock91,Bienstock92}.  See,
e.\,g.,~\cite[\S3.1]{RadovanovicTV21} for graph classes on which the
two-in-a-cycle problem is tractable.

Detecting cycles without the requirement of spanning prespecified
vertices is straightforward.  Even and odd cycles are also long known
to be efficiently detectable
(see,~e.\,g.,~\cite{AlonYZ95,DahlgaardKS17,YusterZ97}).  It takes an
$O(n^2)$-time depth-first search to detect odd cycles even if the
graph is directed (see, e.\,g.,~\cite[Table~1]{BjorklundHK22}).  While
detecting holes (i.\,e., recognizing chordal graphs) is solvable in
$O(n^2)$ time~\cite{RoseTL76,TarjanY84,TarjanY85}, detecting odd
(respectively, even) holes is more difficult. There are early
$O(n^3)$-time algorithms for detecting odd and even holes in planar
graphs~\cite{Hsu87,Porto92}, but the tractability of detecting odd
holes was open for decades (see,
e.\,g.,~\cite{ChudnovskyS18,ConfortiCKV99,ConfortiCLVZ06}) until the
recent major breakthrough of Chudnovsky, Seymour, and
Spirkl~\cite{ChudnovskySSS20-jacm-odd-hole}. Their $O(n^9)$-time
algorithm is later implemented to run in $O(n^8)$ time~\cite{LaiLT20},
immediately implying the best formerly known algorithm for recognizing
perfect graphs based on the Strong Perfect Graph Theorem.  Finding an
odd hole based on Chudnovsky et al.'s $O(n^9)$-time (respectively, Lai
et al.'s $O(n^8)$-time) algorithm for detecting odd holes takes
$O(n^{10})$ (respectively, $O(n^9)$) time.  We improve the time of
detecting and finding odd holes and recognizing perfect graphs to
$O(n^7)$.
\begin{theorem}
\label{theorem:theorem1}
For an $n$-vertex $m$-edge graph $G$, 
\begin{enumerate}[label={(\arabic*)}]
\item it takes $O(mn^5)$ time to either obtain an odd hole of $G$ or
  ensure that $G$ is odd-hole-free and, hence,
\item it takes $O(n^7)$ time to determine whether $G$ is perfect.
\end{enumerate}
\end{theorem}
A shortest cycle of $G$ can be found in $\tilde{O}(n^\omega)$ time
(even if $G$ is directed)~\cite{ItaiR78}. The time becomes $O(n)$ when
$G$ is planar~\cite{ChangL13}.  A shortest odd cycle of $G$ can be
found in $O(n^3)$ time even if $G$ is
directed~(see,~e.\,g.,~\cite[\S1]{BjorklundHK22}).  However, the
previously only known polynomial-time algorithm to find a shortest odd
hole of $G$ takes $O(n^{14})$
time~\cite{ChudnovskySS21-shortest-odd-hole}.  We further reduce the
required time to $O(n^{13})$.
\begin{theorem}
\label{theorem:theorem2}
For an $n$-vertex $m$-edge graph, it takes $O(m^3n^7)$ time to either
obtain a shortest odd hole of $G$ or ensure that $G$ is odd-hole-free.
\end{theorem}
Detecting even cycles in $G$ takes $O(n^2)$ time~\cite{ArkinPY91} and
$O(n^3)$ time even if $G$ is directed~\cite{McCuaig04,RobertsonST99}.
The first polynomial-time algorithm for detecting even holes, running
in about $O(n^{40})$ time
\cite{ConfortiCKV97,ConfortiCKV02a,ConfortiCKV02b}. It takes a line of
intensive efforts in the literature to bring down the complexity to
$O(n^{31})$~\cite{ChudnovskyKS05}, $O(n^{19})$~\cite{daSilvaV13},
$O(n^{11})$~\cite{ChangL15}, and finally $O(n^9)$~\cite{LaiLT20}.  A
shortest even cycle of $G$ is long known to be computable in $O(n^2)$
time~\cite{YusterZ97}.  Very recently, a shortest even cycle of a
directed $G$ is shown to be obtainable in $\tilde{O}(n^{4+\omega})$
time with high probability via an algebraic
approach~\cite{BjorklundHK22}.  On the other hand, the tractability of
finding a shortest even hole, open for $16$
years~\cite{ChudnovskyKS05,Johnson05}, is recently resolved by an
$O(n^{31})$-time algorithm~\cite{CheongL22} which mostly adopts the
$O(n^{31})$-time algorithm for detecting even
holes~\cite{ChudnovskyKS05}.  We show two improved algorithms.  The
less (respectively, more) involved one runs in $O(n^{25})$
(respectively, $O(n^{23})$) time.
\begin{theorem}
\label{theorem:theorem3}
For an $n$-vertex $m$-edge graph $G$, it takes $O(m^7n^9)$ time to
either obtain a shortest even hole of $G$ or ensure that $G$ is
even-hole-free.
\end{theorem}

\subsection{Technical overview and related work}
\label{subsection:subsection1.1}

{\bf Recognizing perfect graphs via detecting odd holes\quad} The
first known polynomial-time algorithm of Chudnovsky, Scott, Seymour,
and Spirkl~\cite{ChudnovskySSS20-jacm-odd-hole} for detecting odd
holes consists of the four subroutines:
\begin{enumerate}[label={(\arabic*)}]
\item Detecting ``jewels'' in $O(n^6)$
  time~\cite[3.1]{ChudnovskyCLSV05}.

\item Detecting ``pyramids'' in $O(n^9)$
  time~\cite[2.2]{ChudnovskyCLSV05}.

\item Detecting ``heavy-cleanable'' shortest odd holes in a graph
  having no jewel and pyramid in $O(n^8)$
  time~\cite[Theorem~2.4]{ChudnovskySSS20-jacm-odd-hole}.

\item Detecting odd holes in a graph having no jewel, pyramid, and
  heavy-cleanable shortest odd hole in~$O(n^9)$
  time~\cite[Theorem~4.7]{ChudnovskySSS20-jacm-odd-hole}.
\end{enumerate}
Lai, Lu, and Thorup~\cite{LaiLT20} improve the complexity to $O(n^8)$
by reducing the time of (2), (3), and (4) to
$\tilde{O}(n^5)$~\cite[Theorem~1.3]{LaiLT20},
$O(n^5)$~\cite[Lemma~6.8(2)]{LaiLT20}, and $O(n^8)$~\cite[Proof of
  Theorem~1.4]{LaiLT20}, respectively.  Finding odd holes based on
Chudnovsky et al.'s $O(n^9)$-time (respectively, Lai et al.'s
$O(n^8)$-time) algorithm for detecting odd holes takes $O(n^{10})$
(respectively, $O(n^9)$) time.  We further improve the time of
detecting and finding odd holes to $O(n^7)$ by the following
arrangement.
\begin{itemize}
\item Extending the concept of a graph containing jewels
  (respectively, heavy-cleanable shortest holes and pyramids) to that
  of a shallow (respectively, medium and deep) graph (defined
  in~\S\ref{section:section2}).
\item Generalizing
\begin{itemize}
\item (1) to an $O(n^7)$-time subroutine for finding a shortest odd
  hole in a shallow graph (Lemma~\ref{lemma:lemma2.2}),
\item (2) to an $\tilde{O}(n^6)$-time subroutine for finding an odd
  hole in a deep graph (Lemma~\ref{lemma:lemma2.1}), and
\item (3) to an $O(n^5)$-time subroutine for finding a shortest odd
  hole in a non-shallow, medium, and non-deep graph
  (Lemma~\ref{lemma:lemma2.3}).
\end{itemize}
\item Specializing
\begin{itemize}
\item (4) to an $O(n^7)$-time subroutine for finding a shortest odd hole in a
non-shallow, non-medium, and non-deep graph (Lemma~\ref{lemma:lemma2.4}).
\end{itemize}
\end{itemize}
Chudnovsky et al.'s $O(n^9)$-time subroutine for (4) has six
procedures. The $i$-th procedures with $i\in\{1,2\}$ (respectively,
$i\in \{3,\ldots,6\}$) enumerate all $O(n^6)$ six-tuples
$x=(x_0,\ldots,x_5)$ (respectively, $O(n^7)$ seven-tuples
$x=(x_0,\ldots,x_6)$) of vertices and spend $O(n^3)$ (respectively,
$O(n^2)$) time for each $x$ to examine whether there is an odd hole of
the $i$-th type that contains all vertices of $x$ other than $x_0$.
Lai et al.'s $O(n^8)$-time subroutine for (4) achieves the improvement
by
\begin{enumerate}[label=(\alph*)]
\item 
reducing the number of enumerated vertices to five and keeping the
examination time in $O(n^3)$ for the $i$-th procedures with
$i\in\{1,3,5\}$ and
\item 
keeping the number of enumerated vertices in six and reducing the
examination time to $O(n^2)$ for the $i$-th procedures with
$i\in\{2,4,6\}$.
\end{enumerate}
Our specialized $O(n^7)$-time subroutine for (4) is based on a new
observation that at most five of the vertices in $x$ suffice for each
of the six procedures to pin down an odd hole.  Skipping a vertex
(i.e., $x_1$ or $x_2$ in the proof of Lemma~\ref{lemma:lemma2.4}) to
reduce the number of rounds from $O(n^6)$ to $O(n^5)$ complicates the
task of examining the existence of an odd hole containing the
remaining five vertices other than $x_0$.  We manage to complete the
task within the same $O(n^2)$ time bound via some data structures.

{\bf Finding a shortest odd hole\quad} Our $O(n^7)$-time algorithm
above is almost one for finding a shortest odd hole.  Among the four
subroutines, only the one for (2) may find a non-shortest odd hole.
Indeed, our $O(n^{13})$-time algorithm for finding a shortest odd hole
is obtained by replacing our subroutine for (2) above with an
$O(n^{13})$-time one for finding a shortest odd hole in a deep and
non-shallow graph (Lemma~\ref{lemma:lemma3.1}), which improves upon
Chudnovsky, Scott, and Seymour's $O(n^{14})$-time
subroutine~\cite[3.2]{ChudnovskySS21-shortest-odd-hole} for finding a
shortest odd hole in a graph containing ``great pyramids'', no
``jewelled'' shortest odd hole, and no $5$-hole.  Chudnovsky et al.'s
subroutine enumerates all $O(n^{12})$ twelve-tuples
$y=(y_0,\ldots,y_{11})$ of vertices and finds for each $y$ in $O(n^2)$
time with the assistance of $(y_0,\ldots,y_4)$ a great pyramid $H$
containing $\{y_5,\ldots,y_{11}\}$.  Specifically, $y_5$ is the
``apex'' of~$H$, $\{y_6,y_7,y_8\}$ forms the ``base" of $H$
(see~\S\ref{section:section2}), and $\{y_9,y_{10},y_{11}\}$ consists
of the interior marker (defined in~\S\ref{subsection:subsection3.1})
vertices of a path of $H$ between its apex and base.  Our improved
$O(n^{13})$-time subroutine is based on a new observation (Claim 1 in
the proof of Lemma~\ref{lemma:lemma3.1}, which strengthens
\cite[7.2]{ChudnovskySS21-shortest-odd-hole}) that a vertex in the
base $\{y_6,y_7,y_8\}$ of $H$ can be omitted in the enumeration,
reducing the number of rounds from $O(n^{12})$ to $O(n^{11})$, without
increasing the time $O(n^2)$ to pin down a shortest odd hole.

{\bf Finding a shortest even hole\quad} Both of our algorithms follow
the approach of Cheong and~Lu~\cite{CheongL22} and Chudnovsky,
Kawarabayashi, and Seymour~\cite{ChudnovskyKS05}, which is different
from that of Lai et al.~\cite[\S6]{LaiLT20} for detecting even holes
in $O(n^9)$ time.  Their $O(n^{31})$-time algorithm consists of the
following subroutines:
\begin{enumerate}[label=(\arabic*)]
\item
Obtaining $O(n^{23})$ sets $X_i\subseteq V(G)$ such that at least one
$G_i=G-X_i$ contains a shortest even hole $C$ of $G$ that is ``neat''
in $G_i$ (i.e., being either good or shallow in~\cite[\S2]{CheongL22},
which roughly means that each shortest $uv$-path of $C$ is a shortest
$uv$-path of $G_i$)~\cite[4.5]{ChudnovskyKS05}.  Specifically,
Chudnovsky et al.~\cite[3.1]{ChudnovskyKS05} show that if a shortest
even hole $C$ of $G$ is not neat in $G$, then $G$ contains ``major''
vertices~\cite[\S2]{ChudnovskyKS05} or ``clear
shortcuts''~\cite[\S3]{ChudnovskyKS05} for $C$.  They obtain $O(n^9)$
subsets $Y_j$ of $V(G)\setminus V(C)$~\cite[2.5]{ChudnovskyKS05} and
$O(n^{14})$ subsets $Z_{j,k}$ of each $V(G-Y_j)\setminus
V(C)$~\cite[4.3 and~4.4]{ChudnovskyKS05}) such that at least one $Y_j$
contains all major vertices of $G$ for $C$ and at least one $Z_{j,k}$
intersects all clear shortcuts of $G-Y_j$ for $C$. Thus, $C$ is neat
in at least one of $G-X_i$ for the $O(n^{23})$ subsets $X_i = Y_j\cup
Z_{j,k}$ of $V(G)$.

\item
Finding in $O(n^8)$ time a shortest even hole in $G_i$ that is neat in
$G_i$~\cite[Lemmas 4 and~5]{CheongL22} (see also
Lemma~\ref{lemma:lemma4.3}).  Specifically, they show that if $C$ is
neat in $G_i$, then $8$ equally spaced vertices of $C$ suffice to pin
down a shortest even hole of $G_i$ via their $O(1)$ pairwise shortest
paths in $G_i$.
\end{enumerate}

Based on an observation of Chang and Lu~\cite[Lemma~3.4]{ChangL15}
(see also Lemma~\ref{lemma:lemma4.5}), we reduce the number of the
above subsets $Y_j$ from $O(n^9)$ to $O(n^3)$~(see
Lemma~\ref{lemma:lemma4.1}), immediately leading to an
$O(n^{25})$-time algorithm for finding a shortest even hole (see
Theorem~\ref{theorem:theorem4} in~\S\ref{section:section4}).  To
further improve the time to $O(n^{23})$, we reduce the number of
vertices for (2) from $8$ to $6$ (see Lemma~\ref{lemma:lemma4.7} and
Figure~\ref{figure:figure5}) based on an observation that the distance
of two far apart vertices of $C$ in $G$ can be bounded.

{\bf Related work\quad} Finding a longest $uv$-path in $G$ that has to
(respectively, need not) be induced is
NP-hard~\cite[GT23]{GareyJ79}~(respectively,~\cite[ND29]{GareyJ79}).
See~\cite{EverettFSMPR97,FonluptU82,Meyniel87} for how an induced even
$uv$-path of $G$ affects the perfection of $G$.
See~\cite{Kriesell01a} for a conjecture by Erd\H{o}s on how an
induced~$uv$-path of~$G$ affects the connectivity between $u$ and $v$
in $G$.  See~\cite{GiacomoLM16,JaffkeKT20} for longest or long induced
paths in special graphs.  The presence of long induced paths in $G$
affects the tractability of coloring $G$~\cite{GaspersHP18}.  See
also~\cite{AbrishamiCPRS21} for the first polynomial-time algorithm
for finding a minimum feedback vertex set of a graph having no induced
path of length at least five.  See~\cite{ChudnovskySS21,CookS22} for
detecting a hole with prespecified parity and length lower bound.
See~\cite{AbrishamiCPRS21,ChudnovskyPPT20} for the first
polynomial-time algorithm for finding an independent set of maximum
weight in a graph having no hole of length at least five.
See~\cite{DalirrooyfardVW19} for upper and lower bounds on the
complexity of detecting an~$O(1)$-vertex induced subgraph.
See~\cite{HoangKSS13} for listing induced paths and holes.
See~\cite[\S4]{ChenF07} for the parameterized complexity of detecting
an induced path of a prespecified length.  See~\cite{CookHPRSSTV21,
  HorsfieldPRSTV22} for determining whether all holes of $G$ have the
same length.

\subsection{Preliminaries and roadmap}
\label{subsection:subsection1.2}
For integers $i$ and $k$, let $[i,k]$ consist of the integers $j$ with
$i\leq j\leq k$ and let $[k]=[1,k]$.  Let $|S|$ denote the cardinality
of a set $S$.  Let $R\setminus S$ for sets $R$ and $S$ consist of the
elements of $R$ that are not in $S$.  Let $E(G)$ for a graph $G$
consist of the edges of $G$ and $\|G\|=|E(G)|$.  A \emph{$k$-graph}
(e.\,g., $2$-path or $5$-hole) is a graph having $k$ edges. A
\emph{triangle} is a 3-cycle.  The \emph{length} of a path or a cycle
is its number of edges.  Let $H\subseteq G$ for a graph $H$ denote
$V(H)\subseteq V(G)$ and $E(H)\subseteq E(G)$.  Let $G-V$ for a set
$V$ of vertices denote $G[V(G)\setminus V]$.  Let $G-v$ for a vertex
$v$ be $G-\{v\}$.  Let $G\setminus E$ for a set $E$ of edges denote
the graph obtained from $G$ by deleting its edges in $E$.  For any
$u\in V(G)$, let $N_G(u)$ consist of the vertices $v$ with $uv\in
E(G)$ and $N_G[u]=\{u\}\cup N_G(u)$.  The \emph{degree} of $u$ in $G$
is $|N_G(u)|$.  A \emph{leaf} of a graph $G$ is a degree-$1$ vertex of
$G$.  Let \emph{$\text{int}(P)$} consist of the interior vertices of a
path $P$.  A \emph{$uv$-path} for vertices $u$ and $v$ is a path with
ends $u$ and $v$.  A \emph{$UV$-path} for vertex sets $U$ and $V$ is a
$uv$-path with $u\in U$ and $v\in V$.  Let $T[u,v]$ with
$\{u,v\}\subseteq V(T)$ for a tree $T$ denote the simple $uv$-path of
$T$.  If vertices $u$ and $v$ of $G$ are connected in $G$, then let
$d_G(u,v)$ denote the length of a shortest $uv$-path of
$G$. Otherwise, let $d_G(u,v)=\infty$.  For any graph $H$, let
$N_G(H)$ consist of the vertices $v\notin V(H)$ with $uv\in E(G)$ for
some $u\in V(H)$ and $N_G[H]=V(H)\cup N_G(H)$.  For any graphs $D$ and
$H$, let $N_G(u,D)=N_G(u)\cap V(D)$ and $N_G(H,D)=N_G(H)\cap V(D)$.
Graphs $H$ and $D$ are \emph{adjacent} (respectively,
\emph{anticomplete}) in $G$ if $N_G(H,D) \neq\varnothing$
(respectively, $N_G[H] \cap V(D) = \varnothing$).

It is convenient to assume that the $n$-vertex $m$-edge graph $G$ of
Theorems~\ref{theorem:theorem1}, \ref{theorem:theorem2},
and~\ref{theorem:theorem3} are connected for the rest of the paper,
which is organized as follows.  Section~\ref{section:section2} proves
Theorem~\ref{theorem:theorem1}.  Section~\ref{section:section3} proves
Theorem~\ref{theorem:theorem2}.  Section~\ref{section:section4} proves
Theorem~\ref{theorem:theorem3}.  Section~\ref{section:section5}
concludes the paper.

\section{Recognizing perfect graphs via detecting odd holes}
\label{section:section2}

The section assumes without loss of generality that $G$ contains no
$5$- or $7$-hole, which can be listed in $O(mn^5)$ time.  A
$D\subseteq V(G)$ with $|D|\leq 5$ is a \emph{spade} for a hole $C$ of
$G$ if (1) $C[D]$ is a $uv$-path, (2) $G[D]$ contains an induced
$uv$-path with length $\|C[D]\|+1$ or $\|C[D]\|-1$, and (3) $C-B$ with
$B= N_G[D\setminus\{u,v\}]\setminus\{u,v\}$ is a shortest $uv$-path of
$G-B$.  A hole $C$ of $G$ is \emph{shallow} if $C$ is a shortest odd
hole of $G$ and there is a spade for $C$.  We comment that a
jewelled~\cite{ChudnovskySS21-shortest-odd-hole} shortest odd hole of
$G$ need not be a shallow hole of $G$ but implies a shallow hole of
$G$.  Let $M_G(C)$ consist of the (major~\cite{ChudnovskyKS05})
vertices $x$ of $G$ such that $N_G(x,C)$ is not contained by any
2-path of $C$.  A hole $C$ of $G$ is \emph{medium} if $C$ is a
shortest odd hole of $G$ and $M_G(C)\subseteq N_G(e)$ holds for an
$e\in E(C)$.  Thus, $5$-holes are medium.  A medium hole is a
heavy-cleanable shortest odd hole
in~\cite{ChudnovskySS21-shortest-odd-hole}.  A triple
$T=(T_1,T_2,T_3)$ of $ab_i$-paths $T_i$ for $i\in[3]$ with
$\|T_1\|<\|T_2\|\leq\|T_3\|$ is a \emph{tripod} of $G$ if $\|T_1\|$ is
minimized over all triples $T$ satisfying the following
\emph{Conditions~Z}:
\begin{enumerate}[label={}, ref={Z}, leftmargin=0pt]
\item 
\label{condition:Z}
\begin{enumerate}[label=\emph{\ref{condition:Z}\arabic*:},ref={\ref{condition:Z}\arabic*},leftmargin=*]
\item \label{Z1} $B(T)=\{b_1,b_2,b_3\}$ induces a triangle of $G$.
\item \label{Z2} $U(T)=T_1\cup T_2\cup T_3$ is an induced tree of
  $G\setminus E(G[B(T)])$ with the leaf set $B(T)$.
\item \label{Z3} $a(T)=a$ is the only degree-3 vertex of $U(T)$.
\item \label{Z4} $C(T)=G[T_2\cup T_3]$ is a shortest odd hole of $G$.
\end{enumerate}
\end{enumerate}
A hole of $G$ is \emph{deep} if it is $C(T)$ for a tripod $T$ of
$G$. Such a $G[U(T)]$ is called an optimal great pyramid of $G$ with
apex $a(T)$ and base $B(T)$ in
\cite{ChudnovskySS21-shortest-odd-hole}.  A graph is \emph{shallow}
(respectively, \emph{medium} and \emph{deep}) if it contains a shallow
(respectively, medium and deep) hole.
\begin{lemma}[{Lai, Lu, and Thorup~\cite[Theorem~1.3]{LaiLT20}}]
\label{lemma:lemma2.1}
It takes $O(mn^4\log^2n)$ time to obtain a $C\subseteq G$ such that
(1) $C$ is an odd hole of $G$ or (2) $G$ is non-deep.
\end{lemma}

\begin{lemma}
\label{lemma:lemma2.2}
It takes $O(mn^{5})$ time to obtain a $C\subseteq G$ such that (1) $C$
is a shortest odd hole of $G$ or (2) $G$ is non-shallow.
\end{lemma}

\begin{lemma}
\label{lemma:lemma2.3}
It takes $O(mn^3)$ time to obtain a $C\subseteq G$ such that (1) $C$
is a shortest odd hole of $G$ or (2) $G$ is shallow, deep, or
non-medium.
\end{lemma} 

\begin{lemma}
\label{lemma:lemma2.4}
It takes $O(mn^5)$ time to obtain a $C\subseteq G$ such that (1) $C$
is a shortest odd hole of $G$ or (2) $G$ is shallow, medium, deep, or
odd-hole-free.
\end{lemma}

Lemma~\ref{lemma:lemma2.2} corresponds to the algorithm for jewelled
shortest odd holes in~\cite[2.1]{ChudnovskySS21-shortest-odd-hole}.
Lemma~\ref{lemma:lemma2.3} improves on the $O(m^2n^4)$-time algorithm
of \cite[6.2]{ChudnovskySS21-shortest-odd-hole}.
Lemma~\ref{lemma:lemma2.4} improves on the $O(m^2n^5)$-time algorithm
of \cite[6.3]{ChudnovskySS21-shortest-odd-hole} and the
$O(m^2n^4)$-time algorithm in \cite[Proof of Theorem~1.4]{LaiLT20}.
We reduce Theorem~\ref{theorem:theorem1} to
Lemmas~\ref{lemma:lemma2.2},~\ref{lemma:lemma2.3},
and~\ref{lemma:lemma2.4} via Lemma~\ref{lemma:lemma2.1}.
Lemmas~\ref{lemma:lemma2.2},~\ref{lemma:lemma2.3},
and~\ref{lemma:lemma2.4} are proved in
\S\ref{subsection:subsection2.1}, \S\ref{subsection:subsection2.2},
and \S\ref{subsection:subsection2.3}.

\begin{proof}[Proof of Theorem~\ref{theorem:theorem1}]
It suffices to prove (1).  It takes $O(m)$ time to determine if one of
the four $C$ is an odd hole of $G$.  If there is one, then (1)
holds. Otherwise, $G$ is non-deep by Lemma~\ref{lemma:lemma2.1},
non-shallow by Lemma~\ref{lemma:lemma2.2}, and non-medium by
Lemma~\ref{lemma:lemma2.3}, implying that $G$ is odd-hole-free by
Lemma~\ref{lemma:lemma2.4}.
\end{proof}

\subsection{Proving Lemma~\ref{lemma:lemma2.2}}
\label{subsection:subsection2.1}

\begin{proof}[Proof of Lemma~\ref{lemma:lemma2.2}]
It takes $O(m)$ time to determine for each $D\subseteq V(G)$ with
$|D|\leq5$ whether $G$ contains odd holes for which $D$ is a spade.
If $G$ contains such odd holes, then let $C_D$ be a shortest of
them. Otherwise, let $C_D=\varnothing$. If all $C_D$ are empty, then
let the $O(mn^5)$-time obtainable $C$ be empty. Otherwise, let $C$ be
a non-empty $C_D$ with minimum $\|C_D\|$.  If $G$ contains a shallow
hole $C^*$, then $0<\|C\|\leq \|C_D\| \leq \|C^*\|$ holds for a spade
$D$ for $C^*$, implying that $C$ is a shortest odd hole of $G$.
\end{proof}

\subsection{Proving Lemma~\ref{lemma:lemma2.3}}
\label{subsection:subsection2.2}

A \emph{clean} hole of $G$ is a medium hole $C$ of $G$ with
$M_G(C)=\varnothing$.

\begin{lemma}
\label{lemma:lemma2.5}
Let $H$ be an induced subgraph of $G$ containing a shortest odd hole
of $G$.
\begin{enumerate}[label=(\arabic*)]
\item If $G$ is non-shallow, then so is $H$.
\item If $G$ is non-deep, then so is $H$.
\end{enumerate}
\end{lemma}

\begin{lemma}[{Chudnovsky, Scott, and Seymour~\cite[Proof of Lemma 6.1]{ChudnovskySS21-shortest-odd-hole}}]
\label{lemma:lemma2.6}
If $u$ and $v$ are vertices of a clean hole $C$ of a non-shallow and
non-deep graph $H$, then the graph obtained from $C$ by replacing the
shortest $uv$-path of $C$ with a shortest $uv$-path of $H$ remains a
clean hole of $H$.
\end{lemma}
We first reduce Lemma~\ref{lemma:lemma2.3} to
Lemma~\ref{lemma:lemma2.5} via Lemma~\ref{lemma:lemma2.6} and then
prove Lemma~\ref{lemma:lemma2.5}
in~\S\ref{subsubsection:subsubsection2.2.1}.  We also include a proof
of Lemma~\ref{lemma:lemma2.6}
in~\S\ref{subsubsection:subsubsection2.2.2} to ensure that it is
implicit in~\cite[Proof of
  Lemma~6.1]{ChudnovskySS21-shortest-odd-hole}.
Lemma~\ref{lemma:lemma2.6} is stronger than
\cite[Theorem~4.1(2)]{ChudnovskyCLSV05} in that $G$ is allowed to
contain jewels or pyramids.  As a matter of fact, the original proof
of \cite[Theorem~4.1(2)]{ChudnovskyCLSV05} already works for
Lemma~\ref{lemma:lemma2.6}: Their careful case analysis shows that if
the resulting subgraph is not a clean hole of $G$, then $G$ contains a
jewel or pyramid. It is not difficult to further infer that each such
jewel (respectively, pyramid) in $G$ contains a shallow (respectively,
deep) hole of $G$.

\begin{proof}[Proof of Lemma \ref{lemma:lemma2.3}]
(Inspired by~\cite[Proof of Lemma 6.8(2)]{LaiLT20}.)
For each $e\in E(G)$ and $u\in V(G)$, spend $O(m)$ time to obtain a
shortest-path tree of $G-N_G(e)\setminus\{u\}$ rooted at $u$, from
which spend $O(n)$ time for each $v\in V(G)$ to obtain a shortest
$uv$-path $P_e(v,u)$, if any, of
$G_e(u,v)=G-(N_G(e)\setminus\{u,v\})$. Let $P_e(u,v)=P_e(v,u)$ for
each $\{u,v\}\subseteq V(G)$ without loss of generality. Thus, it
takes overall $O(mn^3)$ time to obtain for all edges $e$ and distinct
vertices $u$ and $v$ of $G$ with defined $P_e(u,v)$ (i)
$p_e(u,v)=\|P_e(u,v)\|$ and (ii) the neighbor $\phi_e(u,v)$ of $u$ in
$P_e(u,v)$.  Let $p_e(u,v)=\infty$ for undefined $P_e(u,v)$.  Spend
$O(mn^3)$ time to determine if the following equation holds for any
edge $e$ and distinct vertices $b,c,$ and $d$ of $G$:
\begin{equation} \tag{1}
\label{equation:equation3.1}
\begin{aligned}
p_e(c,d) &=3 \\
p_e(c,\phi_e(d,b))&>3\\
p_e(d,\phi_e(c,b))&>3\\
p_e(c,b) &= p_e(d,b) =p_e(c,\phi_e(b,d))-1=p_e(d,\phi_e(b,c))-1.
\end{aligned}
\end{equation}
If Equation~(\ref{equation:equation3.1}) holds for some $(e,b,c,d)$,
then let $C=P_e(b,c)\cup P_e(b,d)\cup P_e(c,d)$ for such an
$(e,b,c,d)$ that minimizes $p_e(b,c)+p_e(b,d)+p_e(c,d)$. Otherwise,
let $C=\varnothing$.
	
We show that if $C^*$ is a medium hole of a non-shallow and non-deep
graph $G$, then $C$ is a shortest odd hole of $G$.  Let $e$ be an edge
of $C^*$ with $M_G(C^*)\subseteq N_G(e)$.  $C^*$ is a clean hole of
the non-shallow and non-deep graph $H=G_e(c,d)$ with
$\{c,d\}=N_{C^*}(e)$ by Lemma~\ref{lemma:lemma2.5}.  For each
$\{u,v\}\subseteq V(C^*)$ such that $\{c,d\}$ is disjoint from the
interior of the shortest $uv$-path of $C^*$, $P_e(u,v)$ is a shortest
$uv$-path of $H$. Therefore, Lemma~\ref{lemma:lemma2.6} implies that
$P_e(b,c)\cup P_e(b,d)\cup P_e(c,d)$ for the $b\in V(C^*)$ with
$d_{C^*}(b,c)=d_{C^*}(b,d)$ is a clean hole of $H$ and hence a
shortest odd hole of $G$.  One can verify from $\|C^*\|\geq9$ that
Equation (\ref{equation:equation3.1}) holds for this $(e,b,c,d)$.
Thus, $C\neq\varnothing$. It remains to show that Equation
(\ref{equation:equation3.1}) for any choice of $(e,b,c,d)$ implies
that $P_e(b,c)\cup P_e(b,d)\cup P_e(c,d)$ is an odd hole of $G$ with
length $p_e(b,c)+p_e(b,d)+p_e(c,d)$: Both $P_e(b,c)$ and $P_e(b,d)$
are induced paths. By
\[
p_e(c,b)= p_e(d,b) = p_e(c,\phi_e(b,d))-1=p_e(d,\phi_e(b,c))-1,
\]
paths $P_e(b,c)-b$ and $P_e(b,d)-b$ are anticomplete. The interior of
$P_e(c,d)$ is anticomplete to $(P_e(c,b)-c)\cup (P_e(d,b)-d)$, since
otherwise $p_e(c,\phi_e(d,b))\leq 3$, $p_e(d,\phi_e(c,b))\leq 3$,
$p_e(c,b)\geq p_e(c,\phi_e(b,d))$, or $p_e(d,b)\geq
p_e(d,\phi_e(b,c))$ holds, violating
Equation~(\ref{equation:equation3.1}).
\end{proof}

\subsubsection{Proving Lemma~\ref{lemma:lemma2.5}}
\label{subsubsection:subsubsection2.2.1}
\begin{proof}[Proof of Lemma \ref{lemma:lemma2.5}]
For the first statement, let $H$ contain a shallow hole $C$ for which
$D$ is a spade, implying that $C$ is a shortest odd hole of $G$.  Let
$C[D]$ be a $uv$-path. $H[D]$ contains an induced $uv$-path $R$ with
length $\|C[D]\|+1$ or $\|C[D]\|-1$.  Hence, $G[D]=H[D]$ contains an
induced $uv$-path $Q\in\{C[D],R\}$ such that the union $C^*$ of
$Q=C[D]$ and a shortest $uv$-path $P$ of
$G-N_G[D\setminus\{u,v\}]\setminus\{u,v\}$ is an odd hole of $G$.
Since $G[D]$ contains an induced $uv$-path, i.e., $C[D]$ or $R$ with
length $\|C^*[D]\|+1$ or $\|C^*[D]\|-1$, $D$ is a spade for $C^*$ in
$G$.  Since $H-N_H[D\setminus\{u,v\}]\setminus\{u,v\}$ is an induced
subgraph of $G-N_G[D\setminus\{u,v\}]\setminus\{u,v\}$, we have
$\|P\|\leq\|C\|-\|C[D]\|$.  By $\|C^*[D]\|\leq\|C[D]\|+1$, $C^*$ is a
shortest odd hole of $G$ and thus a shallow hole of $G$.

For the second statement, let $H$ contain a tripod $T$, implying that
$C(T)$ is a shortest odd hole of $G$. Either $T$ is a tripod of $G$ or
$G$ contains a tripod $T^*$ with $\|T_1^*\|<\|T_1\|$. Thus, $G$ is
deep.
\end{proof}

\subsubsection{Proving Lemma~\ref{lemma:lemma2.6}}
\label{subsubsection:subsubsection2.2.2}

\begin{proof}[Proof of Lemma~\ref{lemma:lemma2.6}]
(Included to ensure that the lemma is implicit
  in~\cite{ChudnovskySS21-shortest-odd-hole}.)
Let $u$ and $v$ be vertices of a clean hole $C$ of a non-shallow and
non-deep graph $H$.
By~\cite[Lemma~4.1]{ChudnovskySS21-shortest-odd-hole}, we have
$d_H(u,v)=d_C(u,v)$.  By~\cite[Lemmas~4.2 and
  4.3]{ChudnovskySS21-shortest-odd-hole}, the graph obtained from $C$
by replacing the shortest $uv$-path of $C$ with a $uv$-path of $H$
with length $d_C(u,v)$ remains a clean hole of $H$.  Hence, the lemma
holds.
\end{proof}

\subsection{Proving Lemma~\ref{lemma:lemma2.4}}
\label{subsection:subsection2.3}

Let $M_G^*(C)=\{x\in M_G(C):|N_G(x,C)|\geq4\}$, whose elements are
called big major vertices for $C$
in~\cite{ChudnovskySS21-shortest-odd-hole}.

\begin{lemma}
\label{lemma:lemma2.7}
A shortest odd hole $C$ of $G$ with $M_G^*(C)\neq M_G(C)$ implies a
tripod $T$ of $G$ with $\|T_1\|=1$.
\end{lemma}

\begin{lemma}
\label{lemma:lemma2.8}
If $C$ is a non-shallow shortest odd hole of $G$, then each $x\in
M_G^*(C)$ admits an $e\in E(C)$ with $M_G^*(C)\subseteq N_G(e)\cup
N_G(x)$.
\end{lemma}

Lemma~\ref{lemma:lemma2.8} is stronger than \cite[Theorem
  5.3]{ChudnovskySS21-shortest-odd-hole} in that $G$ can be shallow.
We first reduce Lemma~\ref{lemma:lemma2.4} to
Lemmas~\ref{lemma:lemma2.7} and~\ref{lemma:lemma2.8} via
Lemmas~\ref{lemma:lemma2.5} and~\ref{lemma:lemma2.6}.  We then prove
Lemmas~\ref{lemma:lemma2.7} and~\ref{lemma:lemma2.8}
in~\S\ref{subsubsubsection:subsubsubsection2.3.1} and
\S\ref{subsubsubsection:subsubsubsection2.3.2}.

\begin{figure}
    \centering
    \tikzstyle{graynode}= [circle,draw=gray,fill=gray, minimum size = 0.1mm]
    \tikzstyle{bluenode}= [circle,draw=blue,fill=blue, minimum size = 0.1mm]
    \tikzstyle{rednode}= [circle,draw=red,fill=red, minimum size = 0.1mm]
    \tikzstyle{blackline}=[-,line width=0.5mm, color=black]
    \tikzstyle{greenline}=[-,line width=0.5mm, color=green]
    \tikzstyle{blueline}=[-,line width=0.5mm, color=blue]
    \tikzstyle{redline}=[-,line width=0.5mm, color=red]
  \begin{tikzpicture}[thick, minimum size = 6mm]
		\node [circle, line width=0.5mm, minimum size=5cm] (C) {};
        \node               (sep)       at (5,0)    []  {};
	    \node               (label)     at (0,-3.25)   []  {(1)};
		\node[bluenode]		(x0)        at (0,0)    [label=below:$x_0$]  {};
		\node[bluenode]		(v0)        at (-0.5,-1)    []  {};
		\node[graynode]		(n1)		at (C.90)	[]	{};
		\node[graynode]		(x2)		at (C.130)	[label=above:$x_2$]	{};
		\node[graynode]		(n2)		at (C.170)	[]	{};
		\node[graynode]		(x1)		at (C.-150)	[label=below:$x_1$]	{};
		\node[graynode]		(n4)		at (C.-110)	[]	{};
		\node[graynode]		(x3)		at (C.-70)	[label=90:$x_3$]	{};
		\node[graynode]		(x4)		at (C.-30)	[label=below:$x_4$]	{};
		\node[graynode]		(x5)		at (C.10)	[label=-135:$x_5$]	{};
		\node[graynode]		(n8)		at (C.50)	[]	{};
		\draw[blackline]    (x1)--(x0)--(x5);
		\draw[blackline]    (n2)--(x0)--(n1);
		\draw[blackline]    (x4)--(v0)--(n4);
		\draw[blackline]    (x2)--(v0)--(n2);
		\begin{pgfonlayer}{background}
		\draw[blueline] (x2) arc (130:210:2.5cm);
		\draw[greenline] (x5) arc (10:130:2.5cm);
		\draw[redline] (x1) arc (-150:10:2.5cm);
		\end{pgfonlayer}
	\end{tikzpicture}
	\begin{tikzpicture}[thick, minimum size = 6mm]
		\node [circle,label={}, line width=0.5mm, minimum size=5cm] (C) {};
	    \node               (label)     at (0,-3.25)   []  {(2)};
		\node[bluenode]		(x0)        at (0,0)    [label=below:$x_0$]  {};
		\node[rednode]		(v1)		at (-3,0.5)	[]	{};
		\node[graynode]		(v2)		at (-1.5,1)	[]	{};
		\node[rednode]		(y)		at (0.5,0.5)	[label=0:$y$]	{};
		\node[rednode]		(n1)		at (C.90)	[]	{};
		\node[graynode]		(x2)		at (C.114)	[label=above:$x_2$]	{};
		\node[graynode]		(n4)		at (C.138)	[]	{};
		\node[graynode]		(n2)		at (C.162)	[]	{};
		\node[rednode]		(x1)		at (C.-174)	[label=-135:$x_1$]	{};
		\node[graynode]		(n5)		at (C.-150)	[]	{};
		\node[graynode]		(n6)		at (C.-126)	[]	{};
		\node[graynode]		(n7)		at (C.-102)	[]	{};
		\node[graynode]		(n8)		at (C.-78)	[]	{};
		\node[graynode]		(n3)		at (C.-54)	[]	{};
		\node[graynode]		(x5)		at (C.-30)	[label=-90:$x_5$]	{};
		\node[graynode]		(x4)		at (C.-6)	[label=-45:$x_4$]	{};
		\node[graynode]		(x3)		at (C.18)	[label=-45:$x_3$]	{};
		\node[graynode]		(n9)		at (C.42)	[]	{};
		\node[graynode]		(n10)		at (C.66)	[]	{};
		\draw[blackline]    (x1)--(x0)--(x4);
		\draw[blackline]    (n3)--(x0)--(n5);)	[]	{};
		\draw[blackline]    (x2)--(y)--(n10);
		\draw[blackline]    (x5)--(y)--(n3);
		\draw[blackline]    (v1)--(n2);
		\draw[blackline]    (n4)--(v2)--(y);
  		\begin{pgfonlayer}{background}
  		\draw[blueline] (x2) arc (114:186:2.5cm);
  		\draw[redline] (x4) arc (-6:114:2.5cm);
 		\draw[greenline] (x1) arc (-174:-6:2.5cm);
  		\end{pgfonlayer}
	\end{tikzpicture}
\caption{(1) An example for the proof of Lemma~\ref{lemma:lemma2.4}
  with $j=1$ and $k=5$. The red arc is the shortest $x_1x_k$-path $P$
  of $H$. The blue arc is $P_1(x_2)$ and the green arc is
  $P_k(x_2)$. $C_1(x_2)=P\cup P_1(x_2)\cup P_k(x_2)$ is a shortest odd
  hole of (1).  (2) An example for the proof of
  Lemma~\ref{lemma:lemma2.4} with $j=2$ and $k=4$. The red arc is the
  shortest $x_2x_k$-path $P$ of $H$. The red vertices denote the
  vertices in $X_1$. Although $y\notin V(H_1)$ and $y\notin I_1$, we
  have $y\in X_1$. Although $y\in V(G_1)$, we have $y\notin
  V(G_0(x_1))$. The blue arc is $P_2(x_1)$ and the green arc is
  $P_k(x_1)$. $C_2(x_1)=P\cup P_2(x_1)\cup P_k(x_1)$ is a shortest odd
  hole of (2).}
\label{figure:figure1}
\end{figure}
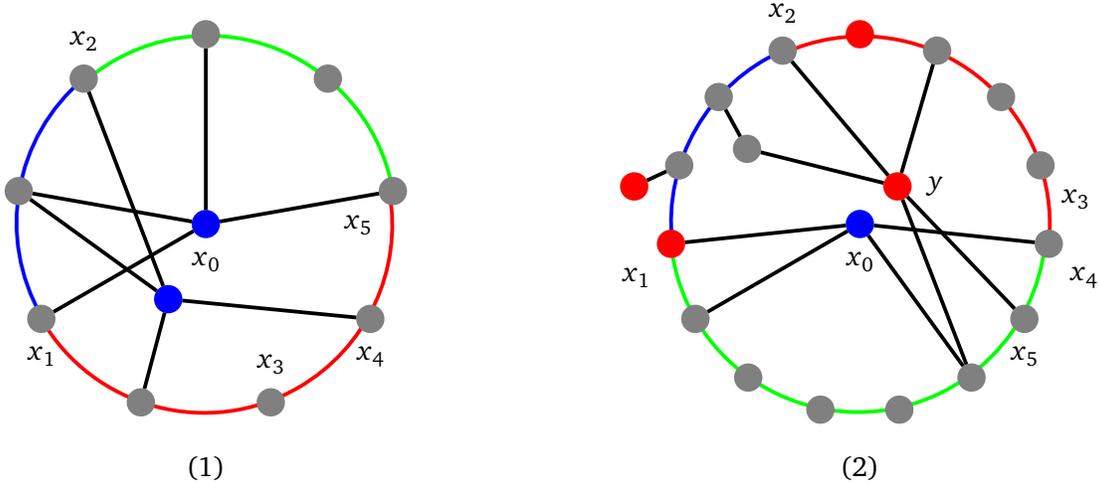

\begin{proof}[Proof of Lemma \ref{lemma:lemma2.4}]
We first show an $O(n^2)$-time two-case subroutine that obtains a
graph for each $k\in[3,5]$ and
\[
 \{x_0,x_j,x_3,x_4,x_5\}\subseteq V(G)
\]
with $j\in[2]$ and $x_4x_5\in E(G)$. If all $O(mn^3)$ of them are
empty, then let the $O(mn^5)$-time obtainable graph $C$ be
empty. Otherwise, let $C$ be a shortest of the nonempty ones. We then
prove that $C$ is a shortest odd hole of a non-shallow, non-medium,
and non-deep $G$ based on the next corollary of
Lemmas~\ref{lemma:lemma2.6} and~\ref{lemma:lemma2.5}: If the shortest
$uv$-path $C^*(u,v)$ of a shortest odd hole $C^*$ of $G$ is contained
by a subgraph $H$ of the non-shallow and non-deep
\[
G^*=G-M_G(C^*),
\]
then each shortest $uv$-path $P$ of $H$ is a shortest path of $G^*$
and we call $H$ a \emph{witness} for~$P$.

Case 1: $j=1.$ Let $P$ be a shortest $x_1x_k$-path of the graph
\[
H=G-(N_G[\{x_0,x_4,x_5\}]\setminus\{x_1,x_3,x_4,x_5\}),
\]
as illustrated by Figure~\ref{figure:figure1}(a).  Let $I$ consist of
the interior vertices of all shortest $x_1x_k$-paths of $H$. Let
\[
G_0=G-((N_G(x_1)\cap N_G(x_k))\cup(N_G[I]\setminus\{x_1,x_k\})).
\]
Spend overall $O(n^2)$ time to obtain for each $i\in \{1,k\}$ and
$v\in V(G_0)$ an arbitrary, if any, shortest $x_iv$-path $P_i(v)$ of
$G_0$ and $R_i(v)=N_{G_0}[P_i(v)-v]$. For each $v\in V(G)$, it takes
$O(n)$ time to determine if
\[
C_1(v)=P\cup P_1(v)\cup P_k(v)
\]
is an odd hole of $G$ via $\|P\|+\|P_1(v)\|+\|P_k(v)\|\equiv1$
$(\text{mod }2)$ and $R_1(v)\cup V(P_k(v))=\{v\}$.  If none of the
$O(n)$ graphs $C_1(v)$ is an odd hole of $G$, then report the empty
graph. Otherwise, report a shortest one of the graphs $C_1(v)$ that
are odd holes.

Case 2: $j=2.$ Let $P$ be a shortest $x_2x_k$-path of the graph
\[
H=G-(N_G[\{x_0,x_4,x_5\}]\setminus\{x_3,x_4,x_5\}),
\]
as illustrated by Figure~\ref{figure:figure1}(b).  Let $I$ consist of
the interior vertices of all shortest $x_2x_k$-paths of $H$. With
\[
H_1=G-(N_G[\{x_0,x_4,x_5\}\cup I]\setminus\{x_2\}),
\]
let $I_1$ consist of the vertices $v$ with $d_{H_1}(v,x_2)\leq\|P\|-3$. With $X_1=N_G(I_1)$ and 
\[
G_1=G-((N_G(X_1)\cap N_G(x_k))\cup(N_G(I_1\cup I)\setminus
(X_1\cup\{x_2,x_3,x_4,x_5\}))),
\]
let each $P_i(v)$ with $i\in\{2,k\}$ and $v\in V(G)$ be a shortest
$x_iv$-path of the graph
\[
G_0(v)=G_1-(X_1\setminus\{v\}).
\]
It takes overall $O(n^2)$ time to determine whether $C_2(v)=P\cup
P_2(v)\cup P_k(v)$ is an odd hole of $G$ for any $v\in V(G)$ using
similar data structures in Case 1. If none of the $O(n)$ graphs
$C_2(v)$ is an odd hole of $G$, then report the empty
graph. Otherwise, report a shortest one of the graphs $C_2(v)$ that
are odd holes.

The rest of the proof shows that the next choice of $j\in
[2],k\in[3,5], x_0\in M_G(C^*)$, and $\{x_1,\dots,x_5\}\subseteq
V(C^*)$ with $x_4x_5\in E(C^*)$ yields a shortest odd hole
$C_j(x_{3-j})$ of $G$: $M_G(C^*)$ is non-empty or else $C^*$ is medium
in $G$. Let $B$ be a longest induced cycle of $G[C^*\cup M_G(C^*)]$
with $|V(B)\cap M_G(C^*)|=1$. Let $B^*=B-x_0$ for the vertex $x_0\in
V(B)\cap M_G(C^*)$. By $M_G(C^*)\not\subseteq N_G(e)$ for any $e\in
E(C^*)$ (or else $C^*$ is medium in $G$), we have
$\|B^*\|\geq3$. Lemmas~\ref{lemma:lemma2.7} and~\ref{lemma:lemma2.8}
imply an $x_4x_5\in E(C^*)$ with
\begin{equation} \tag{1}
\label{equation:equation3.2-dup}
M_G(C^*)\subseteq N_G[\{x_0,x_4,x_5\}]
\end{equation}
Let $k=|V(B^*)\cap \{x_4,x_5\}|+3$. Let $B^*$ (respectively,
$B^*-\{x_4,x_5\})$) be an $x_1x_k$-path (respectively, $x_1x_3$-path)
such that an $x_1x_5$-path of $C^*$ contains $x_3$ and $x_4$. Thus,
$N_G(x_4x_5)\cap V(B^*)\subseteq \{x_3\}$ and $x_1, x_3, x_4$, and
$x_5$ are in order in $C^*$. By maximality of $\|B\|$, we have
\begin{equation} \tag{2}
\label{equation:equation3.3-dup}
M_G(C^*)\subseteq (N_G(x_1)\cap N_G(x_k))\cup N_G(\text{int}(B^*)).
\end{equation}
Let $j\in[2]$ such that $j=1$ if and only if
$\|B^*\|=\|C^*(x_1,x_k)\|$. $B$ is a hole of $G$ shorter than $C^*$ by
$x_0\in M_G(C^*)$, so $\|B^*\|$ is even. Let $x_2$ be the interior
vertex of the non-shortest $x_1x_k$-path of $C^*$ with
\begin{equation} \tag{3}
\label{equation:equation3.4-dup}
\|C^*(x_1,x_2)\|=\|C^*(x_2,x_k)\|-j.
\end{equation}
Thus, $C^*(x_j,x_k)\subseteq H$.  By $x_0\in M_G(C^*)$ and
$\|B^*\|\geq 3$, we have $\|C^*(x_1,x_k)\|\geq3$.  By
Equation~(\ref{equation:equation3.4-dup}), we have
\[
C^*=C^*(x_1,x_k)\cup C^*(x_2,x_k) \cup C^*(x_1,x_2).
\]
Based upon Lemma~\ref{lemma:lemma2.6}, we prove for either case of
$j\in[2]$ that
\[
C_j(x_{3-j})=P\cup P_j(x_{3-j})\cup P_k(x_{3-j})
\]
is a shortest odd hole of $G$ by the following statements via the
aforementioned corollary of Lemmas~\ref{lemma:lemma2.6}
and~\ref{lemma:lemma2.5}.  See Figure~\ref{figure:figure1}.

1. $P$ is a shortest $x_jx_k$-path of $G^*$: By
Equation~(\ref{equation:equation3.2-dup}), we have
$C^*(x_j,x_k)\subseteq H\subseteq G^*$. $H$ is a witness for $P$.

2. Each $P_i(x_{3-j})$ with $i\in\{j,k\}$ is a shortest
$x_ix_{3-j}$-path of $G^*$: If $j=1$, then $B^*=C^*(x_1,x_k).$ By
$\text{int}(B^*)\subseteq I$ and
Equation~(\ref{equation:equation3.3-dup}), we have
$C^*(x_i,x_2)\subseteq G_0\subseteq G^*$ for each $i\in
\{1,k\}$. $G_0$ is a witness for $P_1(x_2)$ and $P_k(x_2)$.  If $j=2$,
then $B^*=C^*(x_1,x_2)\cup C^*(x_2,x_k).$ By
$V(C^*(x_1,x_2)-x_1)\subseteq I_1$ and
$\text{int}(C^*(x_2,x_k))\subseteq I$, we have $x_1\in X_1$ and
$\text{int}(B^*)\subseteq I_1\cup I$. By
Equation~(\ref{equation:equation3.3-dup}), we have $B^*\subseteq G_1$
and $V(G_1)\cap M_G(C^*)\subseteq X_1,$ implying
$C^*(x_1,x_i)\subseteq G_0(x_1)\subseteq G^*$ for each
$i\in\{2,k\}$. $G_0(x_1)$ is a witness for $P_2(x_1)$ and $P_k(x_1)$.
\end{proof}

\subsubsection{Proving Lemma~\ref{lemma:lemma2.7}}
\label{subsubsubsection:subsubsubsection2.3.1}
A path $P$ of $C$ is an \emph{$x$-gap} (see,
e.\,g.,~\cite{ChudnovskySSS20-jacm-odd-hole}) with $x\in M_G(C)$ if
$G[P\cup\{x\}]$ is a hole of $G$ (and thus $\|P\|\geq2)$.  The
shortestness of $C$ implies that each $x$-gap is even.

\begin{proof}[Proof of Lemma \ref{lemma:lemma2.7}]
Let $x\in M_G(C)\setminus M_G^*(C)$, implying that $|N_G(x,C)|\leq3$.
Since $\|C\|$ is odd, there is an edge not in an $x$-gap, implying
that $C[N_G(x,C)]$ contains exactly one edge of $C$. Since $x\in
M_G(C)$, we have $|N_G(x,C)|=3$ and thus $C=C(T)$ for a tripod $T$ of
$G$ with $\|T_1\|=1$.
\end{proof}

\subsubsection{Proving Lemma~\ref{lemma:lemma2.8}}
\label{subsubsubsection:subsubsubsection2.3.2}

A $v\in V(G)$ (respectively, $uv\in E(G)$) is \emph{$X$-complete} with
$X\subseteq V(G)$ if $v\in N_G(x)$ (respectively, $\{u,v\}\subseteq
N_G(x))$ holds for each $x\in X$. Abbreviate $\{x\}$-complete with
$x\in V(G)$ as $x$-complete.  Lemma~\ref{lemma:lemma2.9} is stronger
than \cite[Theorem 5.1]{ChudnovskySS21-shortest-odd-hole} in that $G$
can be shallow.

\begin{lemma}
\label{lemma:lemma2.9}
For any stable $X\subseteq M_G^*(C)$ for a non-shallow shortest odd
hole $C$ of $G$, the number of $X$-complete edges of $C$ is odd.
\end{lemma}
We first reduce Lemma~\ref{lemma:lemma2.8} to
Lemma~\ref{lemma:lemma2.9} and then prove Lemma~\ref{lemma:lemma2.9}.

\begin{proof}[Proof of Lemma \ref{lemma:lemma2.8}]
Assume for contradiction a $G$ with minimum $|V(G)|$ violating the
lemma. We have $M_G^*(C)=V(G)\setminus V(C)$. Let $x_0\in M_G^*(C)$
with $M_G^*(C)\not\subseteq N_G(e)\cup N_G(x_0)$ for each $e\in E(C)$,
which has to be anticomplete to $M_G^*(C)\setminus\{x_0\}$ by
minimality of $|V(G)|$.  Lemma~\ref{lemma:lemma2.9} implies an edge
$x_1x_2$ of $G[M_G^*(C)]$. Minimality of $|V(G)|$ implies for each
$i\in [2]$ an edge $e_i\in E(C)$ that is adjacent to each vertex of
$M_G^*(C)\setminus\{x_{3-i}\}$. Since Lemma~\ref{lemma:lemma2.9}
implies an $\{x_0,x_i\}$-complete edge $f$ of $C$, $G[\{x_0,x_i\}\cup
  e_i]$ is not an induced $x_0x_i$-path $P$ (with $\|P\|=3$) or else
$G[P\cup f]$ contains a 5-hole of $G$.  Thus, each $i\in [2]$ admits
an $\{x_0,x_i\}$-complete end $v_i$ of $e_i$.  By definition of $x_0$,
each $x_i$ with $i\in [2]$ is anticomplete to $e_{3-i}$.  Hence, we
have $v_1\neq v_2$, implying $v_1v_2\in E(C)$ or else
$G[\{x_1,v_1,x_0,v_2,x_2\}]$ is a $5$-hole. However, $e=v_1v_2$ is
adjacent to each member of $M_G^*(C)$: if a $z\in M_G^*(C)$ is
anticomplete to $e$, then $z\notin\{x_0,x_1,x_2\}$ and $z$ is
$\{e_1-v_1,e_2-v_2\}$-complete. Thus, $G[e_1\cup e_2\cup \{z\}]$ is a
$5$-hole, contradiction.
\end{proof}

\begin{proof}[Proof of Lemma \ref{lemma:lemma2.9}]
$C[N_G(x)]$ with $x\in X$ cannot be a $3$-path or else $C$ is a
shallow hole of $G$ with a spade $N_G(x,C)\cup\{x\}$. (This is why
$C$ need be non-shallow). A path $P$ of $C$ is an $xy$-gap with
$\{x,y\}\subseteq X$ and $x\neq y$ if
\begin{itemize}
\item $P$ is an $\{x,y\}$-complete vertex (and thus $\|P\|=0$) or
\item $P$ is a $uv$-path with $N_G(x,P)=\{u\}$ and $N_G(y,P)=\{v\}$ (and thus $\|P\|\geq1$).
\end{itemize}
We first prove the observation that any odd and even $xy$-gaps $P$ and
$Q$ are disjoint and adjacent: $P$ and $Q$ are disjoint or else $P\cup
Q$ contains an odd $x$-gap. Thus,
\[
|N_G(x,P\cup Q)|=|N_G(y,P\cup Q)|=2.
\]
If $P$ were non-adjacent to $Q$, then $G[P\cup Q\cup\{x,y\}]$ is an
odd hole, whose length is $\|C\|$ by shortestness of $C$. By
$\{x,y\}\subseteq X$, the two vertices in $C-V(P\cup Q)$ are
$\{x,y\}$-complete. We have $\|Q\|\neq 0$ or else $C[N_G(x)]$ is a
3-path. Hence, $\|Q\|\geq2$, implying an odd $x$-gap in $C[N_C[Q]]$,
contradiction.

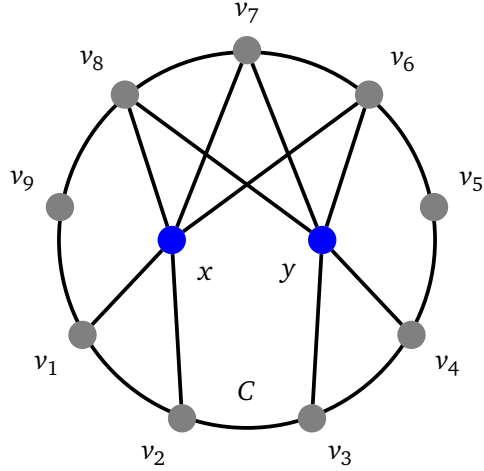
\begin{figure}
    \usetikzlibrary{backgrounds}
    \centering
    \tikzstyle{graynode}= [circle,draw=gray,fill=gray, minimum size = 0.1mm]
    \tikzstyle{bluenode}= [circle,draw=blue,fill=blue, minimum size = 0.1mm]
    \tikzstyle{blackline}=[-,line width=0.5mm, color=black]
    \tikzstyle{greenline}=[-,line width=0.5mm, color=green]
    \tikzstyle{blueline}=[-,line width=0.5mm, color=blue]
    \tikzstyle{redline}=[-,line width=0.5mm, color=red]
  \begin{tikzpicture}[thick, minimum size = 6mm]
		\node [circle,draw=black, line width=0.5mm, minimum size=5cm] (C) {};
		\node               (label)    at (0,-2)  []  {$C$};
		\node[bluenode]		(x)        at (-1,0)    [label=-45:$x$]  {};
		\node[bluenode]		(y)        at (1,0)    [label=-135:$y$]  {};
		\node[graynode]		(v1)		at (C.-150)	[label=-150:$v_1$]	{};
		\node[graynode]		(v2)		at (C.-110)	[label=-110:$v_2$]	{};
		\node[graynode]		(v3)		at (C.-70)	[label=-70:$v_3$]	{};
		\node[graynode]		(v4)		at (C.-30)	[label=-30:$v_4$]	{};
		\node[graynode]		(v5)		at (C.10)	[label=10:$v_5$]	{};
		\node[graynode]		(v6)		at (C.50)	[label=45:$v_6$]	{};
		\node[graynode]		(v7)		at (C.90)	[label=90:$v_7$]	{};
		\node[graynode]		(v8)		at (C.130)	[label=130:$v_8$]	{};
		\node[graynode]		(v9)		at (C.170)	[label=170:$v_9$]	{};
		\draw[blackline]    (v1)--(x)--(v2);
		\draw[blackline]    (v6)--(x)--(v8);
 		\draw[blackline]    (v3)--(y)--(v4);
 		\draw[blackline]    (v6)--(y)--(v8);
 		\draw[blackline]    (x)--(v7)--(y);
	\end{tikzpicture}
\caption{An illustration for the proof of Lemma~\ref{lemma:lemma2.9}
  with $X=\{x,y\}$.  The circle $C$ with vertex set $\{v_1,\dots,
  v_9\}$ is a shortest odd hole of the graph.}
\label{figure:figure2}
\end{figure}

Assume for contradiction that an $X$ with minimum $|X|$ violates the
lemma, implying $|X|\geq 2$.  For each $Y\subseteq X$, let $E_Y$
consist of the $Y$-complete edges of $C$.  By the minimality of $X$,
$|E_X|$ is even and $|E_Y|$ is odd for each $Y\subsetneq X$. Hence,
\[
\sum_{Y\subseteq X, Y\neq\varnothing} |E_Y|
\]
is even.  For each $e\in E(C)$, let $X(e)$ consist of the
$V(e)$-complete vertices of $X$.  We have $e\in E_Y$ if and only if
$Y\subseteq X(e)$ for each $e\in E(C)$ and $Y\subseteq X$.  Therefore,
each edge $e$ of \[Z=\bigcup_{Y\subseteq X, Y\neq\varnothing} E_Y\]
belongs to exactly $2^{|X(e)|}-1$ sets $E_Y$ with nonempty $Y\subseteq
X$.  Hence, \[\sum_{Y\subseteq X, Y\neq\varnothing} |E_Y|=\sum_{e\in
  Z} (2^{|X(e)|}-1)\] is even and thus $|Z|$ is even.  Take
Figure~\ref{figure:figure2} for an example, abbreviating each
$E_{\{v\}}$ with $v\in V(G)$ as $E_v$. $|E_{x}|=|E_{y}|=3$.
$|E_X|=|\{v_6v_7,v_7v_8\}|=2$.
\[
\sum_{Y\subseteq X, Y\neq\varnothing} |E_Y| =|E_{x}|+|E_{y}|+|E_X|=3+3+2=8.
\]
$X(v_1v_2)=\{x\}$.  $X(v_3v_4)=\{y\}$.  $X(v_6v_7)=X(v_7v_8)=X$.  To
see that $e\in E_Y$ if and only if $Y\subseteq X(e)$ for each $e\in
E(C)$ and $Y\subseteq X$, observe for instance that $v_6v_7\in E_{Y}$
for each $Y\subseteq X(v_6v_7)=X$ and $\{x\}\subseteq X(e)$ for each
$e\in E_x$.  To see that each $e\in Z=\{v_1v_2,v_3v_4,v_6v_7,v_7v_8\}$
belongs to exactly $2^{|X(e)|}-1$ sets $E_Y$ with nonempty $Y\subseteq
X$, observe for example that $v_6v_7$ belongs to exactly
$2^{|X(v_6v_7)|}-1 = 3$ sets $E_Y$ with nonempty $Y\subseteq X$, that
is, $E_x$, $E_y$, and $E_X$. Also,
\[
\sum_{Y\subseteq X, Y\neq\varnothing} |E_Y|
= 8=1+1+3+3= 2^{|X(v_1v_2)|}-1+2^{|X(v_3v_4)|}-1+2^{|X(v_6v_7)|}-1+2^{|X(v_7v_8)|}-1
=\sum_{e\in Z} (2^{|X(e)|}-1).
\]
There is an odd $xy$-gap $P$ for an $\{x,y\}\subseteq X$ with $x\neq
y$: The paths in $C\setminus Z$ that is not a vertex can be
partitioned into pairwise edge-disjoint $x$-gaps for an $x\in X$ and
$xy$-gaps for an $\{x,y\}\subseteq X$ with $x\neq y$ via the following
process: for each $u_0u_1\in E(C)\setminus Z$ that is not yet in any
classified $x$-gap or $xy$-gap, let $v_i$ be the vertex of $C-u_{1-i}$
minimizing $d_{C-u_{1-i}}(u_i,v_i)$ such that $N_G(v_i)\cap X\neq
\varnothing$ for each $i\in\{0,1\}$. Let $Q$ be the $v_0v_1$-path of
$C$ containing $u_0u_1$.  If there is a vertex $x\in X$ with
$\{v_0,v_1\}\subseteq N_G(x)$, then classify $Q$ as an
$x$-gap. Otherwise, classify $Q$ as an $xy$-gap for an
$\{x,y\}\subseteq X$ with $x\neq y$.  Therefore, since $\|C\|-|Z|$ is
odd and each $x$-gap is even for each $x\in X$, there is an odd
$xy$-gap $P$ for an $\{x,y\}\subseteq X$ with $x\neq y$.

There is an even $xy$-gap $Q$: Assume for contradiction that all
$xy$-gaps are odd. Thus, $C$ contains no $\{x,y\}$-complete edge,
since an $\{x,y\}$-complete vertex of $C$ is an even $xy$-gap. Hence,
$C$ contains an even number of $x$-complete or $y$-complete edges The
number of edges of $C$ contained by $x$-gaps or $y$-gaps is also
even. Since an edge of $C$ not contained by any $xy$-gaps has to be an
$x$-complete or $y$-complete edge or contained by an $x$-gap or a
$y$-gap, the number of edges in $Q=C-\text{int}(P)$ that are contained
by $xy$-gaps is even. Therefore, then number of $xy$-gaps in $Q$ is
even, implying an $\{x,y\}$-complete end of $Q$, contradicting no even
$xy$-gap in $C$.

By the above observation, $P$ and $Q$ are disjoint and adjacent. Thus,
$R=C-V(P\cup Q)$ is an odd $uv$-path of $C$ with
$\min(|N_G(x,R)|,|N_G(y,R)|)\geq2$.  If $\|R\|=1$, then $R$ is an
$\{x,y\}$-complete edge of $C$.  $\|Q\|\neq0$ or else $C[N_G(x)]$ is a
3-path. By $\|Q\|\geq2$, $C[N_C[Q]]-V(P)$ is an odd $x$-gap or
$y$-gap, contradiction.  If $\|R\|\geq3$, then $S=R-\{u,v\}$ is a path
of $C$ disjoint and nonadjacent to $P\cup Q$.  By the above
observation, $S$ contains no $xy$-gap, implying $N_G(x,S)=\varnothing$
or $N_G(y,S)=\varnothing$. If $N_G(x,S)=\varnothing$ (respectively,
$N_G(y,S)=\varnothing$), then $R$ is an odd $x$-gap (respectively,
$y$-gap), contradiction.
\end{proof}

\section{Finding a shortest odd hole}
\label{section:section3}
By Theorem~\ref{theorem:theorem1}, the section assumes without loss of
generality that $G$ contains odd holes and each odd hole of $G$ has
length at $15$, since all odd holes shorter than $15$ can be listed in
$O(m^3n^7)$ time.
\begin{lemma}
\label{lemma:lemma3.1}
It takes $O(m^3n^7)$ time to obtain a $C\subseteq G$ such that (1) $C$
is a shortest odd hole of $G$ or (2) $G$ contains a shallow hole or no
deep hole.
\end{lemma}
Lemma~\ref{lemma:lemma3.1} improves upon the $O(m^3n^8)$-time
algorithm of \cite[Lemma 3.2]{ChudnovskySS21-shortest-odd-hole}.  We
first reduce Theorem~\ref{theorem:theorem2} to
Lemma~\ref{lemma:lemma3.1} via
Lemmas~\ref{lemma:lemma2.2},~\ref{lemma:lemma2.3},
and~\ref{lemma:lemma2.4} and then prove Lemma~\ref{lemma:lemma3.1}
in~\S\ref{subsection:subsection3.1}.
\begin{proof}[Proof of Theorem \ref{theorem:theorem2}]
Assume for contradiction that none of the four $C\subseteq G$ ensured
by Lemmas~\ref{lemma:lemma3.1}, \ref{lemma:lemma2.2},
\ref{lemma:lemma2.3}, and \ref{lemma:lemma2.4} is a shortest odd hole
of $G$.  By Lemma~\ref{lemma:lemma2.2}, $G$ is non-shallow.  By
Lemma~\ref{lemma:lemma3.1}, $G$ is non-deep.  By
Lemma~\ref{lemma:lemma2.3}, $G$ is non-medium, contradicting
Lemma~\ref{lemma:lemma2.4}.  Thus, it takes $O(m)$ time to obtain a
shortest odd hole of $G$ from the four $C$.
\end{proof}

\subsection{Proving Lemma~\ref{lemma:lemma3.1}}
\label{subsection:subsection3.1}
We call $c=(c_0,\dots,c_4)$ with $\{c_0,\dots,c_4\}\subseteq V(P)$ an
\emph{$\ell$-marker} of a $c_0c_4$-path $P$ if $d_P(c_0,c_2)=
\lceil\|P\|/2\rceil$, $d_P(c_0,c_1)= \min(\ell,d_P(c_0,c_2))$, and
$d_P(c_3,c_4)= \min(\ell,d_P(c_2,c_4))$.  A \emph{$c$-trail} for a
$c=(c_0,\dots,c_k)$ with $\{c_0,\dots,c_k\}\subseteq V(H)$ for a graph
$H$ is the union of one shortest $c_{i-1}c_i$-path of $H$ per
$i\in[k]$.

\begin{lemma}[{Chudnovsky, Scott, and Seymour~\cite[7.1]{ChudnovskySS21-shortest-odd-hole}}]
\label{lemma:lemma3.2}
If $|N_G(x)\cap B(T)|\leq 1$ with $x\in M_G^*(C(T))$ holds for a
tripod $T$ of $G$, then $G[N_G(x,T_1\cup T_i)]$ is an edge for an
$\{i,j\}= \{2,3\}$ with $\|T_j\|\geq3$.
\end{lemma}

\begin{lemma}[{Chudnovsky, Scott, and Seymour~\cite[8.1]{ChudnovskySS21-shortest-odd-hole}}]
\label{lemma:lemma3.3}
If $\{u,v\}\subseteq V(C)$ for a shortest odd hole $C$ of a
non-shallow graph $G$ and $P$ is a $uv$-path of $G$ with $V(P)\cap
M_G^*(C)=\varnothing$ and $\|P\|\leq\|T_1\|$ for a tripod $T$ of $G$,
then $d_C(u,v)\leq \|P\|$.
\end{lemma}

\begin{lemma}[{Chudnovsky, Scott, and Seymour~\cite[4.2]{ChudnovskySS21-shortest-odd-hole}}]
\label{lemma:lemma3.4}
For $\{u,v\}\subseteq V(C)$ with $C=C(T)$ for tripod $T$ of a
non-shallow graph $G$, if $P$ is a $uv$-path of $G$ with $V(P)\cap
M_G^*(C)=\varnothing$ and $\|P\|=d_C(u,v)\leq\|T_1\|$, then the graph
$C'$ obtained from $C$ by replacing the shortest $uv$-path of $C$ with
$P$ remains a shortest odd hole of $G$.
\end{lemma}

\begin{lemma}[{Chudnovsky, Scott, and Seymour~\cite[8.3, 8.4, and 8.5]{ChudnovskySS21-shortest-odd-hole}}]
\label{lemma:lemma3.5}
Let $T$ be a tripod of a non-shallow graph $G$.  Let
$(a,b_i,b_j)=(a(T),T_i[B(T)],T_j[B(T)])$.  For any $\|T_1\|$-marker
$c=(a,c_1,c_2,c_3,b_i)$ of $T_i$ with $\{i,j\}=\{2,3\}$, the triple
obtained from $T$ by replacing $T_i$ with a $c$-trail of the graph
\[
G-((M_G^*(C(T))\cup N_G[T_1-a]\cup N_G[b_j])
\setminus\{a,c_1,c_2,c_3,b_i\})
\]
remains a tripod of $G$.
\end{lemma}
We are ready to prove Lemma~\ref{lemma:lemma3.1} by
Lemmas~\ref{lemma:lemma2.7},~\ref{lemma:lemma2.8},~\ref{lemma:lemma2.9},~\ref{lemma:lemma3.2},~\ref{lemma:lemma3.3},~\ref{lemma:lemma3.4},
and~\ref{lemma:lemma3.5}.

\begin{proof}[Proof of Lemma \ref{lemma:lemma3.1}]
For any $c=(c_0,\dots,c_k)$ with $\{c_0,\dots,c_k\}\subseteq V(H)$ for
a graph $H$, let $P_H(c_0,\dots,c_k)$ be an arbitrary fixed $c$-trail,
if any, of $H$.  For each of the $O(m^2n^7)$ choice of
$\{x,a,c_1,c_2,c_3,d_1,d_2\}\subseteq V(G)$, $\{b,e\}\subseteq E(G)$
with $b=b_2b_3$, and $\{i,j\}=\{2,3\}$, spend $O(m)$ time to determine
whether $G[P_2\cup P_3]$ is an odd hole of $G$, where
\begin{align*}
Y&=(N_G(x)\cup N_G(e))\setminus (V(e)\cup\{d_1,d_2\})\\
G_1 &= G-((Y\cup (N_G[b]\setminus(N_G(b_2)\cap N_G(b_3))))\setminus\{a,b_j\})\\
P_1 &= P_{G_1}(a,b_j)-b_j\\
G_i &= G-((Y\cup N_G[P_1-a]\cup N_G[b_j])\setminus\{a,c_1,c_2,c_3,b_i\})\\
P_i &= P_{G_i}(a,c_1,c_2,c_3,b_i)\\
G_j &= G-(N_G[(P_1\cup P_i)-a]\setminus\{a,b_j\})\\
P_j &= P_{G_j}(a,b_j).
\end{align*}
If there are such odd holes of $G$, then let the $O(m^3n^7)$-time
obtainable $C$ be a shortest of them. Otherwise, let $C=\varnothing$.
We prove that if $T$ is a tripod of a non-shallow $G$, then $C$ is a
shortest odd hole of $G$ by ensuring that $T'''=(P_1,P_2,P_3)$ is a
tripod of $G$ for the following choice of
$\{x,a,c_1,c_2,c_3,d_1,d_2\}\subseteq V(G)$, $\{b,e\}\subseteq E(C^*)$
with $C^*=C(T)$, and $\{i,j\}=\{2,3\}$: For each $i\in\{2,3\}$, let
$b_i=T_i[B(T)]$ and $C_i=G[T_1\cup T_i]$.  Let
$(a,b)=(a(T),b_2b_3)$. Thus, $a\notin N_G(b)$.  We first claim the
following statements:
	
\emph{Claim 1}: There are $\{x,d_1,d_2\}\subseteq V(G)$,
$\{i,j\}=\{2,3\}$, and $e\in E(C^*)$ with $M_G^*(C^*)\subseteq Y$ and
$Y \cap V(C_i-\{a,b_i\})= \varnothing$.

\emph{Claim 2}: $(P_1,T_2,T_3)$ is a tripod of $G$.

By Claim 1, let $\{c_1,c_2,c_3\}\subseteq V(T_i)$ so that
$c=(a,c_1,c_2,c_3,b_i)$ is a $\|T_1\|$-marker of $T_i$.  By Claim 2,
$\|P_1\|=\|T_1\|$ and $\text{int}(T_i)\cap N_G[P_1-a]=\varnothing$. By
$M_G^*(C^*)\subseteq Y$ and $Y \cap V(C_i-\{a,b_i\})= \varnothing$, we
have
\[
T_i\subseteq G_i\subseteq G-((M_G^*(C(T))\cup N_G[T_1-a]\cup N_G[b_j])
\setminus\{a,c_1,c_2,c_3,b_i\}).
\]
Thus, $\|P_i\|\leq \|T_i\|$. By Lemma~\ref{lemma:lemma3.5}, the triple
$T''$ obtained from $T'$ by replacing $T_i$ with $P_i$ is a tripod of
$G$, implying $\text{int}(T_j)\cap N_G[(P_1\cup P_i)-a]=\varnothing$.
Hence, $T_j\subseteq G_j$, implying that $\|P_j\|\leq\|T_j\|$.  By
definition of $G_j$, $\text{int}(P_j)\cap N_G[(P_1\cup
  P_i)-a]=\varnothing$.  By $a\notin N_G(b)$, $C'=G[P_i\cup P_j]$ is a
hole with $\|C'\|\leq\|C^*\|$. If $\|C'\|<\|C^*\|$, then a $G[P_1\cup
  P_2\cup P_3]-V(P_k-a)$ with $k\in[3]$ is an odd hole of $G$ shorter
than $C^*$ by $1\leq \|P_k\|\leq \|T_k\|$ for each $k\in[3]$,
contradiction. Thus, $C'$ is a shortest odd hole of $G$, implying that
$T'''$ is a tripod of $G$. The rest of the proof shows Claims 1 and 2.
	
We first prove Claim 1.  Assume an $x\in M_G^*(C)\setminus N_G(b)$ or
else the claim holds with $(x,e)=(b_2,b)$, $d_1=T_1[B]$, and $d_2\in
N_{T_2}(b_2)$. Lemma~\ref{lemma:lemma2.8} implies an $e\in E(C)$ with
minimum $k=|V(e)\cap \{a,b_2,b_3\}|$ such that $M_G^*(C)$ is contained
by the set
\[
N=(N_G(x)\cup N_G(e))\setminus V(e).
\]
Thus, $x\in N_G(e)$. Lemma~\ref{lemma:lemma3.2} implies an
$\{i,j\}=\{2,3\}$ with $\|T_j\|\geq3$ such that $G[N_G(x,C_i)]$ for
the hole $C_i=G[T_1\cup T_i]$ is an edge~$e_i$. Assume for
contradiction at least three vertices in the set
\[
I=N\cap (V(C_i)\setminus\{a,b_i\})=(V(e_i)\cup
N_G(e,C_i))\setminus\{a,b_i\}.
\]
We have $e\notin E(T_i)$ or else $x\in N_G(e)$ implies $|I|\leq2$. By
$x\in N_G(e)\setminus N_G(b)$, we have $e\neq b$, implying $e\in
E(T_j)$. We have $V(e)\not\subseteq \text{int}(T_j)$ or else
$N_G(e,C_i)\subseteq\{a\}$ implies $|I|\leq2$. Thus, $e=uv$ with $u\in
\text{int}(T_j)$ and $v\in \{a,b_j\}$. If $v=a$, then $v\notin N_G(x)$
or else $|I|=|N_{C_i}(a)|=2$. If $v=b_j$, then $v\notin N_G(x)$ by
$x\notin N_G(b)$. By $\|T_j\|\geq3$, the neighbor $w$ of $u$ in
$T_j-v$ is in $\text{int}(T_j)$. By minimality of $k$, there is a
vertex
\[
y\in M_G^*(C)\setminus((N_G(x)\cup N_G(uw))\setminus\{u,w\}),
\]
implying that $P=xuvy$ is an induced 3-path of
$G$. Lemma~\ref{lemma:lemma2.9} implies an $\{x,y\}$-complete edge $f$
of $C$. Thus $G[f\cup P]$ contains a 5-hole of $G$,
contradiction. Hence, Claim 1 holds with $\{d_1,d_2\}=I$.

It remains to prove Claim 2.  Let $b_1$ be the end of $P_1$ with
$b_1b_j\subseteq P_{G_1}(a,b_j)$. $G[\{b_1,b_2,b_3\}]$ is a triangle
or else $b_1\in N_G[b]$ contradicts $b_1\in V(G_1)$. By
\[
V(T_1)\cap(N_G[b]\setminus(N_G(b_2)\cap N_G(b_3)))=Y \cap
V(C_i-\{a,b_i\})=\varnothing,
\]
we have $G[V(T_1)\cup\{b_j\}]\subseteq G_1$, implying
\[
1\leq \|P_1\|\leq\|T_1\|
\]
via $a\notin N_G(b)$.  Thus, Claim 2 follows from the next claim,
which implies ($V(P_1-a)\cap V(T_k-a)=\varnothing$ and)
$\|P_1\|=\|T_1\|$ by the minimality of $\|T_1\|$.

\emph{Claim 3}: Each $D_k=G[P_1\cup T_k]$ with $k\in\{2,3\}$ is a hole of $G$. 

To prove Claim 3, assume for contradiction an $e_k=u_1u_k\in
E(G)\setminus\{b_1b_k\}$ for a $k\in\{2,3\}$ with $u_1\in V(P_1-a)$
and $u_k\in V(T_k-a)$ (without ruling out the possibility of $e_k\in
E(D_k)$) that minimizes
\[
d(e_k)=n\cdot d_{P_1}(u_1,b_1)+d_{T_k}(u_k,b_k),
\]
as illustrated by Figure~\ref{figure:figure3}.
	
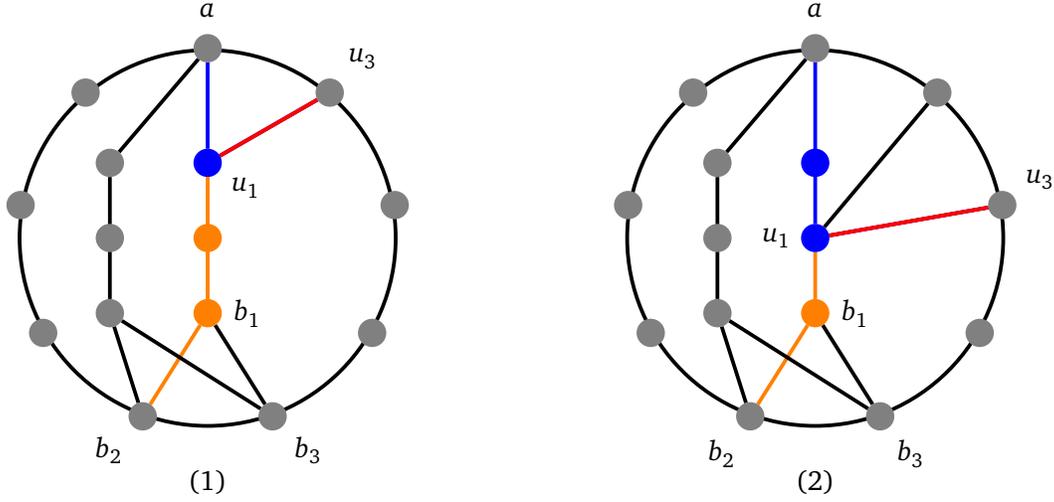
\begin{figure}
\centering
        \tikzstyle{graynode}= [circle,draw=gray,fill=gray, minimum size = 0.1mm]
        \tikzstyle{bluenode}= [circle,draw=blue,fill=blue, minimum size = 0.1mm]
        \tikzstyle{orangenode}= [circle,draw=orange,fill=orange, minimum size = 0.1mm]
        \tikzstyle{blackline}=[-,line width=0.5mm, color=black]
        \tikzstyle{greenline}=[-,line width=0.5mm, color=green]
        \tikzstyle{orangeline}=[-,line width=0.5mm, color=orange]
        \tikzstyle{blueline}=[-,line width=0.5mm, color=blue]
        \tikzstyle{redline}=[-,line width=0.5mm, color=red]
    	\begin{tikzpicture}[thick, minimum size = 6mm]
    		\node [circle,draw=black, line width=0.5mm, minimum size=5cm] (C) {};
            \node               (sep)       at (5,0)    []  {};
    	    \node               (label)     at (0,-3.25)   []  {(1)};
    		\node[orangenode]		(b1)        at (0,-1) [label=0:$b_1$] {};
    		\node[orangenode]		(t1)        at (0,0) [] {};
    		\node[bluenode]		(u1)        at (0,1) [label=-5:$u_1$] {};
    		\node[graynode]		(p1)        at (-1.3,-1) [] {};
    		\node[graynode]		(p2)        at (-1.3,0) [] {};
    		\node[graynode]		(p3)        at (-1.3,1) [] {};
    		\node[graynode]		(v1)		at (C.-150)	[]	{};
    		\node[graynode]		(b2)		at (C.-110)	[label=-135:$b_2$]	{};
    		\node[graynode]		(b3)		at (C.-70)	[label=-45:$b_3$]	{};
    		\node[graynode]		(v4)		at (C.-30)	[]	{};
    		\node[graynode]		(v6)		at (C.10)	[]	{};
    		\node[graynode]		(u3)		at (C.50)	[label=50:$u_3$]	{};
    		\node[graynode]		(a)		at (C.90)	[label=90:$a$]	{};
    		\node[graynode]		(v8)		at (C.130)	[]	{};
    		\node[graynode]		(v9)		at (C.170)	[]	{};
    		\begin{pgfonlayer}{background}
    		\draw[orangeline] (b2)--(b1)--(u1)--(u3);
    		\draw[blueline]  (a)--(u1)--(u3);
    		\draw[blackline] (b2)--(p1);
    		\draw[blackline] (b1)--(b3)--(p1)--(p3)--(a);
    		\draw[redline] (u1)--(u3);
            \end{pgfonlayer}
    	\end{tikzpicture}
        \begin{tikzpicture}[thick, minimum size = 6mm]
    		\node [circle,draw=black, line width=0.5mm, minimum size=5cm] (C) {};
    	    \node               (label)     at (0,-3.25)   []  {(2)};
    		\node[orangenode]		(b1)        at (0,-1) [label=0:$b_1$] {};
    		\node[bluenode]		(u1)        at (0,0) [label=180:$u_1$] {};
    		\node[bluenode]		(t1)        at (0,1) [] {};
    		\node[graynode]		(p1)        at (-1.3,-1) [] {};
    		\node[graynode]		(p2)        at (-1.3,0) [] {};
    		\node[graynode]		(p3)        at (-1.3,1) [] {};
    		\node[graynode]		(v1)		at (C.-150)	[]	{};
    		\node[graynode]		(b2)		at (C.-110)	[label=-110:$b_2$]	{};
    		\node[graynode]		(b3)		at (C.-70)	[label=-70:$b_3$]	{};
    		\node[graynode]		(v4)		at (C.-30)	[]	{};
    		\node[graynode]		(u3)		at (C.10)	[label=10:$u_3$]	{};
    		\node[graynode]		(v6)		at (C.50)	[]	{};
    		\node[graynode]		(a)		at (C.90)	[label=90:$a$]	{};
    		\node[graynode]		(v8)		at (C.130)	[]	{};
    		\node[graynode]		(v9)		at (C.170)	[]	{};
    		\begin{pgfonlayer}{background}
    		\draw[orangeline] (b2)--(b1)--(u1)--(u3);
    		\draw[blueline]  (a)--(u1)--(u3);
    		\draw[blackline] (b2)--(p1);
    		\draw[blackline] (b1)--(b3)--(p1)--(p3)--(a);
    		\draw[redline] (u1)--(u3);
    		\draw[blackline] (u1)--(v6);
            \end{pgfonlayer}
    	\end{tikzpicture}
\caption{Illustrations for Claim 3 in the proof of
  Lemma~\ref{lemma:lemma3.1} with $k=3$.  The circle is a deep hole
  $C^*$ of $G$.  $Q$ is the concatenation of the red edge $u_1u_3$ and
  the blue path.  $R$ is the concatenation of $u_1u_3$ and the orange
  path.  (1) shows Equation~(\ref{equation:equation4.2}), since
  $G[P_1\cup T_k]-a$ is an odd hole of $G$ shorter than $C^*$.  (2)
  reflects Claim 3, since $G[P_1[u_1,b_1]\cup T_k[u_k,b_k]]$ is an odd
  hole of $G$ shorter than $C^*$.}
\label{figure:figure3}
\end{figure}
We deduce a contradiction via the following statements:

\begin{enumerate}[label={}, ref={S}, leftmargin=0pt]
\item 
\label{condition:S}
\begin{enumerate}[label=\emph{\ref{condition:S}\arabic*:},ref={\ref{condition:S}\arabic*},leftmargin=*]
\item \label{S1} $e_k\notin E(D_k)$.
\item \label{S2} $u_1\neq b_1$.
\item \label{S3} $d_{T_k}(a,u_k)\neq d_{P_1}(a,u_1)+1$.
\item \label{S4} If $\{u,v\}\subseteq V(C^*)$ and $P$ is a $uv$-path
  of $G$ with $\text{int}(P)\subseteq V(P_1)$ and $\|P\|\leq\|T_1\|$,
  then $d_{C^*}(u,v)\leq \|P\|$.
\end{enumerate}
\end{enumerate}

By Statements~\ref{S1} and \ref{S2}, the (unnecessarily induced)
$au_k$-path $Q=P_1[a,u_1]\cup e_k$ of $G$ satisfies
$\|Q\|\leq\|P_1\|\leq \|T_1\|$.  By Statement~\ref{S4},
\[
d_{C^*}(a,u_k)\leq \|Q\|=d_{P_1}(a,u_1)+1,
\]
implying $d_{C^*}(a,u_k)=d_{T_k}(a,u_k)$ by $\|Q\|\leq\|P_1\|\leq
\|T_1\|<\|T_\ell\|+\|T_k[u_k,b_k]\|+1$ with $\ell=5-k$.  Thus, by
Statement~\ref{S3},
\begin{equation} \tag{1}
\label{equation:equation4.1}
d_{C^*}(a,u_k)=d_{T_k}(a,u_k)\leq d_{P_1}(a,u_1).
\end{equation}

We next prove
\begin{equation} \tag{2}
\label{equation:equation4.2}
d_{P_1}(u_1,b_1)+2\leq \|T_1\|.
\end{equation}
Assume $\|T_1\|\leq d_{P_1}(u_1,b_1)+1$ for contradiction. By
$\|P_1\|\leq \|T_1\|$, we have $\|P_1\|=\|T_1\|$ and $au_1\in E(P_1)$.
By Equation~(\ref{equation:equation4.1}), $au_k\in E(P_k)$. By the
choice of $(u_1,u_k)$ and Statement~\ref{S2}, $G[P_1\cup T_k]-a$ is a
hole of $G$ with length
\[
\|P_1\|+\|T_k\|=\|T_1\|+\|T_k\|<\|C^*\|.
\]
Thus, $\|T_1\|+\|T_k\|$ is even, implying that $C_k$ is an odd hole
shorter than $C^*$, contradiction.  See Figure~\ref{figure:figure3}(1)
as an illustration.

Consider the (unnecessarily induced) $u_kb_\ell$-path $R=e_k\cup
P_1[u_1,b_1]\cup b_1b_\ell$ of $G$.  By
Equation~(\ref{equation:equation4.2}) and Statement~\ref{S4}, we have
$d_{C^*}(u_k,b_\ell)\leq \|R\|=d_{P_1}(u_1,b_1)+2$ and
$d_{C^*}(u_k,b_\ell)=d_{T_k}(u_k,b_k)+1$ by
$\|R\|\leq\|T_1\|<\|T_\ell\|+d_{T_k}(a,u_k)$.  Thus,
\begin{equation} \tag{3}
\label{equation:equation4.3}
d_{T_k}(u_k,b_k)-1\leq d_{P_1}(u_1,b_1).
\end{equation}

Combining Equations~(\ref{equation:equation4.1})
and~(\ref{equation:equation4.3}), we have
\[
\|T_k\|-1\leq\|P_1\|\leq\|T_1\|\leq\|T_k\|-1
\]
and thus $d_{T_k}(u_k,b_k)-1= d_{P_1}(u_1,b_1)$. By the choice of
$(u_1,u_k)$ and Statement S2, $G[P_1[u_1,b_1]\cup T_k[u_k,b_k]]$ is a
hole of $G$ shorter than $C^*$ and it is odd by $d_{T_k}(u_k,b_k)-1=
d_{P_1}(u_1,b_1)$, contradiction.  See Figure~\ref{figure:figure3}(2)
as an illustration.

It remains to show Statement~\ref{S1},~\ref{S2},~\ref{S3},
and~\ref{S4} via the following equation:
\begin{equation} \tag{4}
\label{equation:equation4.4}
M_G^*(C^*) \cap V(P_1)= \varnothing,
\end{equation}
which is immediate from $M_G^*(C^*)\subseteq Y$ and the definition of
$G_1$.

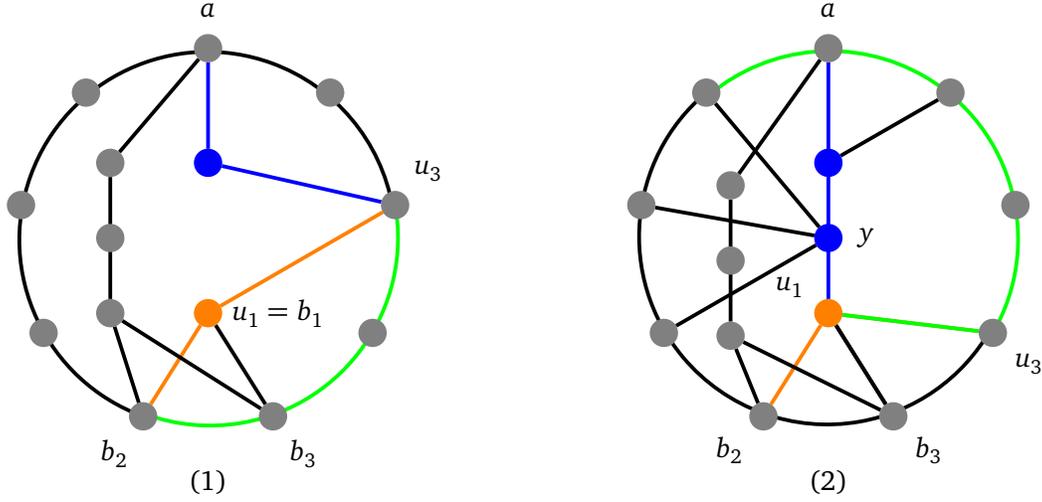
\begin{figure}
        \centering
        \tikzstyle{graynode}= [circle,draw=gray,fill=gray, minimum size = 0.1mm]
        \tikzstyle{bluenode}= [circle,draw=blue,fill=blue, minimum size = 0.1mm]
        \tikzstyle{orangenode}= [circle,draw=orange,fill=orange, minimum size = 0.1mm]
        \tikzstyle{blackline}=[-,line width=0.5mm, color=black]
        \tikzstyle{greenline}=[-,line width=0.5mm, color=green]
        \tikzstyle{orangeline}=[-,line width=0.5mm, color=orange]
        \tikzstyle{blueline}=[-,line width=0.5mm, color=blue]
        \tikzstyle{redline}=[-,line width=0.5mm, color=red]
    	\begin{tikzpicture}[thick, minimum size = 6mm]
    		\node [circle, line width=0.5mm, minimum size=5cm] (C) {};
            \node               (sep)       at (5,0)    []  {};
    	    \node               (label)     at (0,-3.25)   []  {(1)};
    		\node[orangenode]		(u1)        at (0,-1) [label=right:{$u_1=b_1$}] {};
    		\node[bluenode]		(t2)        at (0,1) [] {};
    		\node[graynode]		(p1)        at (-1.3,-1) [] {};
    		\node[graynode]		(p2)        at (-1.3,0) [] {};
    		\node[graynode]		(p3)        at (-1.3,1) [] {};
    		\node[graynode]		(v1)		at (C.-150)	[]	{};
    		\node[graynode]		(b2)		at (C.-110)	[label=-110:$b_2$]	{};
    		\node[graynode]		(b3)		at (C.-70)	[label=-70:$b_3$]	{};
    		\node[graynode]		(v4)		at (C.-30)	[]	{};
    		\node[graynode]		(u3)		at (C.10)	[label=50:$u_3$]	{};
    		\node[graynode]		(v6)		at (C.50)	[]	{};
    		\node[graynode]		(a)		at (C.90)	[label=90:$a$]	{};
    		\node[graynode]		(v8)		at (C.130)	[]	{};
    		\node[graynode]		(v9)		at (C.170)	[]	{};
    		\begin{pgfonlayer}{background}
    		\draw[blueline]  (u3)--(t2)--(a);
    		\draw[orangeline](b2)--(u1)--(u3);
    		\draw[blackline] (u1)--(b3)--(p1)--(b2);
    		\draw[blackline] (p1)--(p3)--(a);
		    \draw[greenline] (u3) arc (10:-110:2.5cm);
		    \draw[blackline] (b2) arc (250:10:2.5cm);
            \end{pgfonlayer}
    	\end{tikzpicture}
    	\begin{tikzpicture}[thick, minimum size = 6mm]
    		\node [circle, line width=0.5mm, minimum size=5cm] (C) {};
    	    \node               (label)     at (0,-3.25)   []  {(2)};
    		\node[orangenode]		(u1)        at (0,-1) [label=170:$u_1$] {};
    		\node[bluenode]		(y)        at (0,0) [label=0:$y$] {};
    		\node[bluenode]		(t1)        at (0,1) [] {};
    		\node[graynode]		(p1)        at (-1.3,-1.3) [] {};
    		\node[graynode]		(p2)        at (-1.3,-0.3) [] {};
    		\node[graynode]		(p3)        at (-1.3,0.7) [] {};
    		\node[graynode]		(v1)		at (C.-150)	[]	{};
    		\node[graynode]		(b2)		at (C.-110)	[label=-135:$b_2$]	{};
    		\node[graynode]		(b3)		at (C.-70)	[label=-45:$b_3$]	{};
    		\node[graynode]		(u3)		at (C.-30)	[label=-30:$u_3$]	{};
    		\node[graynode]		(v4)		at (C.10)	[]	{};
    		\node[graynode]		(v6)		at (C.50)	[]	{};
    		\node[graynode]		(a)		at (C.90)	[label=90:$a$]	{};
    		\node[graynode]		(z)		at (C.130)	[]	{};%
    		\node[graynode]		(v9)		at (C.170)	[]	{};
    		\begin{pgfonlayer}{background}
    		\draw[orangeline] (b2)--(u1)--(u3);
    		\draw[blueline]  (a)--(u1);
    		\draw[blackline] (v6)--(t1);
    		\draw[blackline] (b3)--(u1);
    		\draw[blackline] (u1)--(b3)--(p1)--(p3)--(a);
    		\draw[blackline] (p1)--(b2);
    		\draw[greenline] (u1)--(u3);
    		\draw[blackline] (y)--(v1);
    		\draw[blackline] (z)--(y)--(v9);
    		\draw[blackline] (z) arc (130:330:2.5cm);
		    \draw[greenline] (u3) arc (-30:130:2.5cm);
            \end{pgfonlayer}
        \end{tikzpicture}
\caption{The circle in each example is a deep hole $C^*$ of the
  graph. (1) An illustration for Statement~\ref{S1} in the proof of
  Lemma~\ref{lemma:lemma3.1}.  The orange path is shorter than the
  green arc, violating Lemma~\ref{lemma:lemma3.3}.  (2) An
  illustration for Statement~\ref{S2} in the proof of
  Lemma~\ref{lemma:lemma3.1}. The green path is an odd $y$-gap.}
\label{figure:figure4}
\end{figure}

Statement~\ref{S1}: Assume $e_k\in E(D_k)$ for contradiction.  By
$b_1\in N_G[b]$ and $V(T_k)\cup N_G[b]=\varnothing$, we have
$b_1\notin V(T_k)$, implying $u_1\notin V(T_k)$ and $u_k\in V(P_1)$ by
the choice of $(u_1,u_k)$.  By Equation~(\ref{equation:equation4.4})
and Lemma~\ref{lemma:lemma3.3}, we have
\[
d_{C^*}(a,u_k)\leq d_{P_1}(a,u_k),
\]
implying $d_{C^*}(a,u_k)=d_{T_k}(a,u_k)$ by
$d_{P_1}(a,u_k)<\|T_\ell\|+1+d_{T_k}(b_k,u_k)$.  Hence,
$d_{P_1}(a,u_k) \geq d_{T_k}(a,u_k)$, implying
\[
d_{P_1}(u_k,b_1)\leq d_{T_k}(u_k,b_k)-1
\]
by $\|P_1\|\leq\|T_1\|\leq\|T_k\|-1$.  By
Equation~(\ref{equation:equation4.4}) and Lemma~\ref{lemma:lemma3.3},
we have
\[
d_{C^*}(u_k,b_\ell)\leq \|P_1[u_k,b_1]\cup
b_1b_\ell\|=d_{P_1}(u_k,b_1)+1,
\]
implying $d_{C^*}(u_k,b_\ell)=d_{T_k}(u_k,b_k)+1$ by
$\|P_1[u_k,b_1]\cup b_1b_\ell\|<\|T_\ell\|+d_{T_k}(a,u_k)$.  Hence, we
have $d_{T_k}(u_k,b_k)\leq d_{P_1}(u_k,b_1)$, contradicting
$d_{P_1}(u_k,b_1)\leq d_{T_k}(u_k,b_k)-1$.  See
Figure~\ref{figure:figure4}(1) as an illustration.

Statement~\ref{S2}: Assume $u_1=b_1$ for contradiction, implying
$\{u_k,b_2,b_3\}\subseteq N_G(u_1,C^*)$. By
Equation~(\ref{equation:equation4.4}), $|N_G(u_1,C^*)|\leq3$, implying
$N_G(u_1,C^*)=\{u_k,b_2,b_3\}$. By $u_1=b_1\notin N_G(a)$,
$\|T_1\|\geq \|P_1\|\geq2$ and thus
\begin{equation} \tag{5}
\label{equation:equation4.2.5}
M_G(C^*) =M_G^*(C^*)
\end{equation}
by Lemma~\ref{lemma:lemma2.7}.  By
Equations~(\ref{equation:equation4.4})
and~(\ref{equation:equation4.2.5}), $C^*[N_G(u_1,C^*)]$ is the 2-path
$b_\ell b_ku_k$.  Thus,
\[
C_k^*=G[V(C^*-b_k)\cup\{u_1\}]
\]
is a shortest odd hole of $G$.  There is a vertex $y\in
M_G^*(C_k^*)\cap V(P_1)$ or else Lemma~\ref{lemma:lemma3.3} with
$\|P_1\|\leq\|T_1\|$ implies $d_{C_k^*}(a,u_1)\leq\|P_1\|$,
contradicting $\|P_1\|<\|T_k\|$ and $\|P_1\|<\|T_\ell\|+1$.
Equations~(\ref{equation:equation4.4})
and~(\ref{equation:equation4.2.5}) imply $yu_1\in E(P_1)$ and that
\[
G[N_G(y,C_k^*-u_1)]=G[N_G(y,C^*-b_k)]
\]
is a 2-path nonadjacent to $u_1$ or else $G[y\cup N_G(y,C_k^*)]$ is a
spade for $C_k^*$, contradicting that $G$ is non-shallow. Thus,
$C_k^*$ contains an odd $y$-gap, contradiction.  See
Figure~\ref{figure:figure4}(2) as an illustration.

Statement~\ref{S3}: Assume for contradiction $d_{T_k}(a,u_k)=
d_{P_1}(a,u_1)+1$, implying $d_{C^*}(a,u_k)=d_{T_k}(a,u_k)$ by
$d_{P_1}(a,u_1)+1<\|T_\ell\|+\|T_k[u_k,b_k]\|+1$.  By
Statements~\ref{S1} and~\ref{S2}, we have $\|P_1[a,u_1]\cup e_k\|\leq
\|P_1\|\leq \|T_1\|$, implying a $(Q_1,Q_2,Q_3)$ with
\[
(Q_1,\{Q_2,Q_3\})=(P_1[u_1,b_1],\{P_1[a,u_1] \cup T_\ell,e_k\cup
T_k[u_k,b_k]\})
\]
satisfying all Conditions~\ref{condition:Z} by
Equation~(\ref{equation:equation4.4}) and Lemma~\ref{lemma:lemma3.4},
contradicting the minimality of $\|T_1\|$.

Statements~\ref{S4} follows from Equation~(\ref{equation:equation4.4})
and Lemma~\ref{lemma:lemma3.3}.
\end{proof}

\section{Finding a shortest even hole}
\label{section:section4}
By Lai, Lu, and Thorup's $O(m^2n^5)$-time algorithm for detecting even
holes~\cite[Theorem~1.6]{LaiLT20} and the fact that all even holes
shorter than $24$ can be listed in $O(m^6n^{11})$ time, the section
assumes without loss of generality that $G$ contains even holes and
the length of a shortest even hole of $G$ is at least $24$.  We first
prove the following weaker version of Theorem~\ref{theorem:theorem3}
in \S\ref{subsection:subsection4.1} and then prove
Theorem~\ref{theorem:theorem3} in \S\ref{subsection:subsection4.2}.
\begin{theorem}
\label{theorem:theorem4}
It takes $O(mn^{23})$ time to 
obtain a shortest even hole of $G$.
\end{theorem}

\subsection{Our first improved algorithm for finding a shortest even hole}
\label{subsection:subsection4.1}
Most of the definitions below are adopted or adjusted
from~\cite{CheongL22,ChudnovskyKS05}.  Let $C$ be a shortest even hole
of $G$.  Let $J_G(C)$ consist of the (major) vertices $x\in V(G)$ such
that $N_G(x,C)$ contains three or more vertices that are pairwise
non-adjacent in $G$.  $C$ is \emph{clear}~(see,
e.\,g.,~\cite[\S6.3]{LaiLT20}) in $G$ if $J_G(C)=\varnothing$.

\begin{lemma}
\label{lemma:lemma4.1}
It takes $O(m^{1.5}n^2)$ time to obtain $O(mn)$ sets $X_i\subseteq
V(G)$ such that a shortest even hole of $G$ is clear in at least one
$G-X_i$.
\end{lemma}

Let $P$ be a $uv$-path of $G$ for distinct and nonadjacent vertices
$u$ and $v$ of $C$.  $P$ is \emph{$C$-bad} (called $C$-bad
$C$-shortcut in~\cite{CheongL22}) if (i) the union of $P$ and any
$uv$-path of $C$ is not a shortest even hole of $G$, (ii)
$2\leq\|P\|\leq d_C(u,v)$, and (iii) $\|P\|<\frac{\|C\|}{4}$.  $P$ is
\emph{$C$-worst} if it is $C$-bad and either (i) $\|P\| = \|P'\|$ and
$d_C(u,v) \geq d_C(u',v')$ or (ii) $\|P\| < \|P'\|$ holds for each
$C$-bad $u'v'$-path $P'$ in $G$.  $P$ is \emph{$C$-flat} (called
$C$-shallow in~\cite{CheongL22}) if it is $C$-bad, $\|P\| \geq d_C(u,
v) - 1$, and $G[P\cup C_2]$ for the $uv$-paths $C_1$ and $C_2$ of $C$
with $\|C_1\|\leq\|C_2\|$ is a hole.  Observe that $G$ contains a
$C$-worst path if and only if it contains a $C$-bad path.

\begin{lemma}[{Chudnovsky, Kawarabayashi, and Seymour~\cite[Proof of 4.5]{ChudnovskyKS05}}]
\label{lemma:lemma4.2}
It takes $O(n^{16})$ time to obtain $O(n^{14})$ sets $X_i\subseteq
V(G)$ such that for each clear shortest even hole $C$ of $G$, there is
an $X_i$ with (i) $C\subseteq G-X_i$ and (ii) every $C$-worst path in
$G-X_i$ is $C$-flat.
\end{lemma}

A graph $G$ is \emph{flat} if there is a shortest even hole $C$ of $G$
such that $G$ contains (i) a $C$-flat $C$-worst path or (ii) no
$C$-bad path.
\begin{lemma}[{Cheong and Lu~\cite[Lemmas~4 and~5]{CheongL22}}] 
\label{lemma:lemma4.3}
It takes $O(n^8)$ time to obtain a $C\subseteq G$ such that (1) $C$ is a shortest even hole of 
$G$ or (2) $G$ is not flat.
\end{lemma}

We comment that
Lemma~\ref{lemma:lemma4.2} is implicit in~\cite[Proof of 4.5]{ChudnovskyKS05}.
Their algorithm consists of three phases.
If we view phases~2 and~3 as one piece, then it goes like
\begin{enumerate}[label={(\arabic*)}]
    \item getting $O(n^9)$ subsets $A_1,\ldots,A_a$ of $V(G)$ and
    \item getting $b_i=O(n^{14})$ subsets $B_{i,1}\ldots B_{i,b_i}$ of $V(G-X_i)$  for each $i\in[a]$.
\end{enumerate}
They first show that a shortest even hole $C$ of $G$ is clear in
$G-A_r$ for at least one $r\in[a]$.  The majority of the proof shows
that every $C$-worst path in $G-A_r-B_{r,i}$ is $C$-flat for at least
one $i\in[b_r]$ based upon the fact that $C$ is clear in $G-A_r$.
Lemma~\ref{lemma:lemma4.2} follows by substituting its $(G,X_i)$ with
this $(G-A_r,B_{r,i})$.  We first reduce
Theorem~\ref{theorem:theorem4} to Lemma~\ref{lemma:lemma4.1} via
Lemmas~\ref{lemma:lemma4.2} and~\ref{lemma:lemma4.3} and then prove
Lemma~\ref{lemma:lemma4.1}
in~\S\ref{subsubsection:subsubsection4.1.1}.

\begin{proof}[Proof of Theorem~\ref{theorem:theorem4}]
By Lemmas~\ref{lemma:lemma4.1} and~\ref{lemma:lemma4.2}, it takes
$O(mn^{17})$ time to obtain $O(mn^{15})$ sets $X_i\subseteq V(G)$ such
that a $G-X_i$ contains a shortest even hole $C$ of $G$ and each
$C$-worst path in $G-X_i$ is $C$-flat.  Thus, at least one $G-X_i$ is
flat and contains a shortest even hole of $G$.  A shortest even hole
of $G$ can be obtained by applying Lemma~\ref{lemma:lemma4.3} on all
$G-X_i$ in $O(mn^{23})$ time.
\end{proof}

\subsubsection{Proving Lemma~\ref{lemma:lemma4.1}}
\label{subsubsection:subsubsection4.1.1}
\begin{lemma}[{da~Silva and Vu\v{s}kovi\'{c}~\cite[\S2]{daSilvaV07}}]
\label{lemma:lemma4.4}
It takes $O(m^{1.5}n^2)$ time to obtain all $O(n^2)$ maximal cliques
of a $4$-hole-free graph $G$.
\end{lemma}

\begin{lemma}[{Chang and Lu~\cite[Lemma~3.4]{ChangL15}}]
\label{lemma:lemma4.5}
If $C$ is a shortest even hole of a $4$-hole-free graph $G$, then
either $J_G(C)\subseteq N_G(v)$ holds for a vertex $v$ of $C$ or
$G[J_G(C)]$ is a clique.
\end{lemma}

\begin{proof}[Proof of Lemma~\ref{lemma:lemma4.1}]
Apply Lemma~\ref{lemma:lemma4.4} to obtain all $O(n^2)$ maximal
cliques of $G$ in $O(m^{1.5}n^2)$ time.  The $O(n^3)$-time obtainable
$O(mn)$ subsets of $V(G)$ are
\begin{enumerate}[label=(\arabic*)]
\item the $O(mn)$ sets $N_G(v)\setminus \{u,w\}$ for all $2$-paths
  $uvw$ of $G$ and
\item the $O(n^2)$ sets $V(K)$ for all maximal cliques $K$ of $G$.
\end{enumerate}
Let $C$ be a shortest even hole of $G$.  If $J_G(C)\subseteq N_G(v,C)$
for a $v\in V(C)$, then $C$ is clear in $G-(N_G(v)\setminus \{u,w\})$
for the vertices $u$ and $w$ such that $uvw$ is a $2$-path of $C$.
Otherwise, $J_G(C)\subseteq K$ holds for a maximal clique $K$ of $G$
with $V(K)\cap V(C)=\varnothing$ by Lemma~\ref{lemma:lemma4.5},
implying that $C$ is clear in $G-V(K)$.
\end{proof}

\subsection{Our second improved algorithm for finding a shortest odd hole}
\label{subsection:subsection4.2}

\begin{lemma}[{Chudnovsky, Kawarabayashi, Seymour~\cite[\S4]{ChudnovskyKS05}}]
\label{lemma:lemma4.6}
It takes $O(m^8n^3)$ time to obtain $O(m^7n^3)$ induced subgraphs
$G_i$ of $G$ such that at least one $G_i$ is flat and contains a
shortest even hole of $G$.
\end{lemma}

\begin{lemma}
\label{lemma:lemma4.7}
It takes $O(n^6)$ time to obtain a $C\subseteq G$ such that $C$ is a
shortest even hole of $G$ or $G$ is not flat.
\end{lemma}

Lemma~\ref{lemma:lemma4.6} improves upon the $O(n^{25})$-time
algorithm of \cite[4.5]{ChudnovskyKS05} (see
also~\cite[Lemma~3]{CheongL22}) that produces $O(n^{23})$ subgraphs.
Lemma~\ref{lemma:lemma4.7} improves upon the $O(n^8)$-time algorithm
of Lemma~\ref{lemma:lemma4.3}.

\begin{proof}[Proof of Theorem~\ref{theorem:theorem3}]
Immediate by Lemmas~\ref{lemma:lemma4.6} and~\ref{lemma:lemma4.7}.
\end{proof}
It remains to prove Lemmas~\ref{lemma:lemma4.6}
and~\ref{lemma:lemma4.7} in~\S\ref{subsubsection:subsubsection4.2.1}
and~\S\ref{subsubsection:subsubsection4.2.2}, respectively.
\subsubsection{Proving Lemma~\ref{lemma:lemma4.6}}
\label{subsubsection:subsubsection4.2.1}

We include a proof of Lemma~\ref{lemma:lemma4.6} to ensure that it is
implicit in \cite[\S4]{ChudnovskyKS05}.  The proof is adjusted
from~\cite[4.5]{ChudnovskyKS05} with the modification of replacing
their~(2.5) by Lemma~\ref{lemma:lemma4.1}.  Let $C$ be a shortest even
hole of $G$.  A $C$-bad $P=up_1\cdots p_kv$ is \emph{$C$-clear} if
\begin{itemize}[leftmargin=*]
\item there is no $C$-flat $P'$ with $\text{int}(P)=\text{int}(P')$,
\item $N_G(p_i,C)=\varnothing$ for each $i$ with $1<i<k$, and
\item $C[N_G(p_1,C)]$ and $C[N_G(p_k,C)]$ are vertex disjoint
  $t$-paths for a $t\in[2]$.
\end{itemize}
Let $P=up_1\ldots p_kv$ be $C$-clear.  Let $P$ be of \emph{type~$t$}
with $t\in[2]$ if $C[N_G(p_1)]$ is a $t$-path.
\begin{lemma}[{Chudnovsky, Kawarabayashi, and Seymour~\cite[4.4]{ChudnovskyKS05}}]
\label{lemma:lemma4.8}
It takes $O(m^2n^3)$ time to obtain $O(m^2n^2)$ subsets $Y_j$ such
that if $C$ is a clear shortest even hole of $G$ that contains a
$C$-clear $C$-worst path of type~$1$, then at least one $G-Y_j$
contains $C$ and no $C$-clear path of type~$1$.
\end{lemma}

\begin{lemma}[{Chudnovsky, Kawarabayashi, and Seymour~\cite[4.3]{ChudnovskyKS05}}]
\label{lemma:lemma4.9}
It takes $O(m^4n)$ time to obtain $O(m^4)$ subsets $Y_j$ such that if
$C$ is a clear shortest even hole of $G$ that contains a $C$-clear
$C$-worst path of type~$2$, then at least one of $G-Y_j$ contains $C$
and no $C$-clear path of type~$2$.
\end{lemma}

\begin{lemma}[{Chudnovsky, Kawarabayashi, and Seymour~\cite[3.1]{ChudnovskyKS05}}]
\label{lemma:lemma4.10}
If a path $P$ is $C$-worst for a shortest even hole $C$ of $G$, then
(i) $P$ is $C$-clear or $C$-flat or (ii) $\text{int}(P)$ consists of a
single vertex of $J_G(C)$.
\end{lemma}

\begin{proof}[Proof of Lemma~\ref{lemma:lemma4.6}]
The algorithm starts with applying Lemma~\ref{lemma:lemma4.1} on $G$
to obtain $O(mn)$ sets $X_i\subseteq V(G)$ in $O(m^{1.5}n^2)$ time
such that a shortest even hole $C$ of $G$ is clear in at least one
$G_i=G-X_i$.  It then runs the following symmetric steps.
\begin{enumerate}[label=(\arabic*), leftmargin=*]
\item
Apply Lemma~\ref{lemma:lemma4.8} on each $G_i$ to obtain $O(m^2n^2)$
sets $Y_{ij}$ in $O(m^2n^3)$ time such that if $C$ is clear in $G_i$
that contains a $C$-clear $C$-worst path of type~$1$, then at least
one of $G_{ij}=G_i-Y_{ij}$ contains $C$ and no $C$-clear path of
type~$1$. Apply Lemma~\ref{lemma:lemma4.9} on each $G_{ij}$ to obtain
$O(m^4)$ sets $Z_{ijk}$ in $O(m^4n)$ time such that if $C$ is clear in
$G_{ij}$ that contains a $C$-clear $C$-worst path of type~$2$, then at
least one of $G_{ijk}=G_{ij}-Z_{ijk}$ contains $C$ and no $C$-clear
path of either type.

\item
Apply Lemma~\ref{lemma:lemma4.9} on each $G_i$ to obtain $O(m^4)$ sets
$Y_{ij}$ in $O(m^4n)$ time such that if $C$ is clear in $G_i$ that
contains a $C$-clear $C$-worst path of type~$2$, then at least one of
$G_{ij}=G_i-Y_{ij}$ contains $C$ and no $C$-clear path of type~$2$.
Apply Lemma~\ref{lemma:lemma4.8} on each $G_{ij}$ to obtain
$O(m^2n^2)$ sets $Z_{ijk}$ in $O(m^2n^3)$ time such that if $C$ is
clear in $G_{ij}$ that contains a $C$-clear $C$-worst path of
type~$1$, then at least one of $G_{ijk}=G_{ij}-Z_{ijk}$ contains $C$
and no $C$-clear path of either type.
\end{enumerate}
The $O(m^8n^3)$-time obtainable $O(m^7n^3)$ graphs are the above
$G_i$, $G_{ij}$, and $G_{ijk}$.  Assume for contradiction that each of
the above graphs that contains a shortest even hole of $G$ is not
flat.  Lemma~\ref{lemma:lemma4.1} implies a $G_i$ that contains a
shortest even hole $C$ of $G$ such that $C$ is clear in $G_i$. Thus,
$G_i$ is not flat, implying that a $C$-worst $uv$-path $P$ in $G_i$ is
not $C$-flat. By Lemma~\ref{lemma:lemma4.10} and
$J_{G_i}(C)=\varnothing$, $P$ is $C$-clear. Let $P$ be of type $t$
with $t\in[2]$.  Lemmas~\ref{lemma:lemma4.8} and~\ref{lemma:lemma4.9}
imply a $G_{ij}$ that contains $C$ and no $C$-clear path of type $i$
and thus a $G_{ijk}$ that contains $C$ and no $C$-clear path of either
type.  By Lemma~\ref{lemma:lemma4.10} and $J_{G_{ijk}}(C)=\varnothing$
and the fact that $G_{ijk}$ contains no $C$-clear path, each $C$-worst
path of $G_{ijk}$ is $C$-flat. Thus, $G_{ijk}$ is flat, contradicting
the initial assumption.
\end{proof}

\subsubsection{Proving Lemma~\ref{lemma:lemma4.7}}
\label{subsubsection:subsubsection4.2.2}

Let $P(u,v)$ for each $\{u,v\}\subseteq V(G)$ be an arbitrary but
fixed shortest $uv$-path of $G$, if there is one.  A $3$-tuple
$r=(r_1,r_2,r_3)\in\{0,1\}^3$ is \emph{nice} if $r_1\geq r_2\geq r_3$.
Let $a\in S_r(a_1,a_2,a_3,a_4,a_5)$ for a nice $r$ if there is an
integer $\ell\geq 2$ with
\begin{alignat*}{10}
    & d_G(a_1,a_2) \quad &&= \quad 2\ell+r_1+r_3 \\
    & d_G(a_2,a_3) \quad &&= \quad 2\ell+r_2 \\
    & d_G(a_1,a_3) \quad &&\geq \quad 2\ell+r_1+r_3 \\
    & d_G(a,a_1) \quad &&= \quad 2\ell+r_1 \\
    & d_G(a,a_2) \quad &&\geq \quad 2\ell+r_1+r_3 \\
    & d_G(a,a_3) \quad &&= \quad 2\ell+r_2+r_3 \\
    & d_G(a,a_4) \quad &&= \quad \ell\\
    & d_G(a,a_5) \quad &&= \quad \ell+r_2\\
    & d_G(a_1,a_4) \quad &&= \quad \ell+r_1\\
    & d_G(a_3,a_5) \quad &&= \quad \ell+r_3.
\end{alignat*}
A $6$-tuple $(a_1,a_2,a_3,b_1,b_2,b_3)$ is \emph{$r$-valid} for a nice
$r$ if there is an integer $\ell\geq 2$ such that the following
statements hold with $P_1=P(a_2,b_1)$, $P_2=P(b_1,a_1)$,
$P_3=P(a_1,b_2)$, and $P_4=P(a_3,b_3)$.
\begin{itemize}[leftmargin=*]
\item The length of $P_1$ (respectively, $P_2$, $P_3$ and $P_4$) is
  $\ell+r_1$ (respectively, $\ell+r_3$, $\ell+r_1$, and $\ell+r_3$),
\item $G[P_1\cup P_2]$, $G[P_2\cup P_3]$ are both paths,
\item $P_1$ and $P_3$ are anticomplete, and
\item $P_i$ and $P_4$ are anticomplete for each $i\in[3]$.
\end{itemize}
A shortest even hole $C$ is \emph{good} in $G$ if $G$ contains no
$C$-bad path.  A shortest even hole $C$ is \emph{flat} in $G$ if every
$C$-worst path in $G$ is $C$-flat.
\begin{lemma}
\label{lemma:lemma4.11}
Let $T=(a_1,a_2,a_3,b_1,b_2,b_3)$ be $r$-valid for a nice
$r=(r_1,r_2,r_3)$.  If sets $S_r(a_1,a_2,a_3,b_2,b_3)$ and
$S_r(b_1,b_2,b_3,a_2,a_3)$ are both nonempty, then $G$ contains an
even hole of length $8\ell+2(r_1+r_2+r_3)$, where
$\ell=d_G(a_1,b_1)-r_3$.
\end{lemma}

\begin{lemma}
\label{lemma:lemma4.12}
If $G$ has a good shortest even hole, then there is an $r$-valid
$T=(a_1,a_2,a_3,b_1,b_2,b_3)$ for a nice $r$ such that both
$S_r(a_1,a_2,a_3,b_2,b_3)$ and $S_r(b_1,b_2,b_3,a_2,a_3)$ are
nonempty.
\end{lemma}

\begin{lemma}[{Cheong and Lu~\cite[Lemma~4]{CheongL22}}]
\label{lemma:lemma4.13}
For any $n$-vertex graph G, it takes $O(n^6)$ time to obtain a
$C\subseteq G$ such that (i) $C$ is a shortest even hole of $G$ or
(ii) $G$ contains no flat shortest even hole.
\end{lemma}
We first reduce Lemma~\ref{lemma:lemma4.7} to
Lemmas~\ref{lemma:lemma4.11} and~\ref{lemma:lemma4.12} via
Lemma~\ref{lemma:lemma4.13} and then prove
Lemmas~\ref{lemma:lemma4.11} and~\ref{lemma:lemma4.12}.
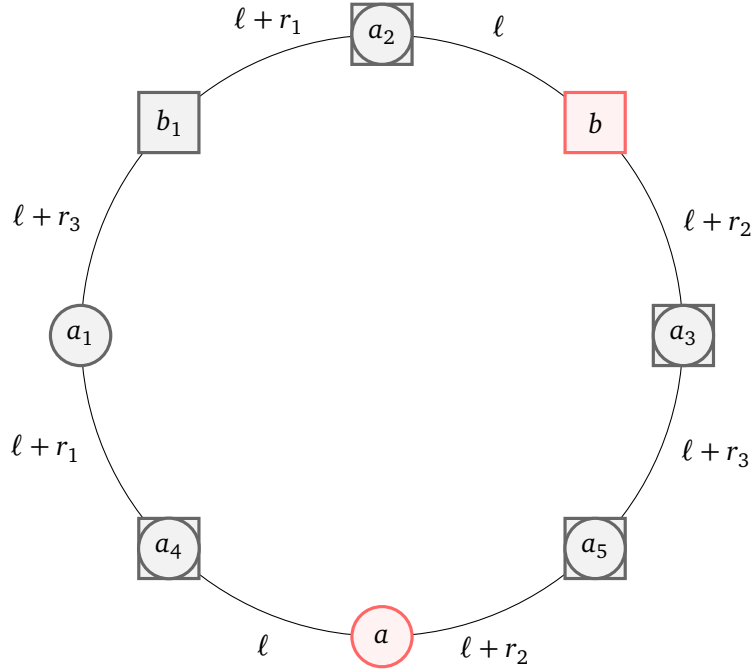
\begin{figure}
\centering
\begin{tikzpicture}
[
typeA/.style={circle, draw=red!60, fill=red!5, very thick, minimum size=8mm},
typeB/.style={rectangle, draw=red!60, fill=red!5, very thick, minimum size=8mm},
typeA1/.style={circle, draw=black!60, fill=black!5, very thick, minimum size=8mm},
typeB1/.style={rectangle, draw=black!60, fill=black!5, very thick, minimum size=8mm},
]
\node [circle, draw, minimum size=8cm] (C) {};
\node[typeB](b) at (C.45) {$b$};
\node[typeB1](b1) at (C.135) {$b_1$};
\node[typeB1](b2a4) at (C.225) {};
\node[typeB1](b3a5) at (C.315) {};
\node[typeB1](b4a2) at (C.90) {};
\node[typeB1](b5a3) at (C.0) {};

\node[typeA](a) at (C.270) {$a$};
\node[typeA1](a1) at (C.180) {$a_1$};
\node[typeA1](a2) at (C.90) {$a_2$};
\node[typeA1](a3) at (C.0) {$a_3$};
\node[typeA1](a4) at (C.225) {$a_4$};
\node[typeA1](a5) at (C.315) {$a_5$};

\node[](t1) at (C.22) {};
\node[](t2) at (C.67) {};
\node[](t3) at (C.112) {};
\node[](t4) at (C.157) {};
\node[](t5) at (C.202) {};
\node[](t6) at (C.247) {};
\node[](t7) at (C.292) {};
\node[](t8) at (C.337) {};
\node[right = 1pt of t1](a3b) {$\ell+r_2$};
\node[above = 1pt of t2](ba2) {$\ell$};
\node[above = 1pt of t3](a2b1) {$\ell+r_1$};
\node[left = 1pt of t4](b1a1) {$\ell+r_3$};
\node[left = 1pt of t5](a1a4) {$\ell+r_1$};
\node[below = 1pt of t6](a4a) {$\ell$};
\node[below = 1pt of t7](aa5) {$\ell+r_2$};
\node[right = 1pt of t8](a5a3) {$\ell+r_3$};
\end{tikzpicture}
\caption{We guess the $6$ black vertices on $C$ that form a valid
  $6$-tuple. Let $a_2=b_4$, $a_3=b_5$, $a_4=b_2$ and $a_5=b_3$. The
  existence of $a$ and $b$ are ensured by $S_r$ on the $5$ circles and
  rectangles respectively.}
\label{figure:figure5}
\end{figure}

\begin{proof}[Proof of Lemma~\ref{lemma:lemma4.7}]
We first compute the following items  in $O(n^6)$ time.
\begin{itemize}
\item For each $4$-tuple $(u_1,u_2,v_1,v_2)$, an indicator for whether
  $G[P(u_1,v_1)\cup P(u_2,v_2)]$ is a path, and an indicator for
  whether $P(u_1,v_1)$ and $P(u_2,v_2)$ are anticomplete.
\item For each nice $r$ and $5$-tuple $(a_1,a_2,a_3,a_4,a_5)$, an
  indicator whether $S_r(a_1,a_2,a_3,a_4,a_5)$ is non-empty.
\end{itemize}
Phase~$1$ checks for each nice $r=(r_1,r_2,r_3)$ and each $6$-tuple
$T=(a_1,a_2,a_3,b_1,b_2,b_3)$ of vertices whether $T$ is $r$-valid,
$S_r(a_1,a_2,a_3,b_2,b_3)\ne\varnothing$, and
$S_r(b_1,b_2,b_3,a_2,a_3)\ne\varnothing$.  If so, then record $8\ell +
2(r_1+r_2+r_3)$ with $\ell=d_G(a_1,b_1)-r_3$.  Maintain the minimum
$m$ of recorded value.  Phase~$2$ runs Lemma~\ref{lemma:lemma4.13} and
records $\|C\|$.

We can assume $G$ to be flat, implying a shortest even hole $C$ that
is good or flat in $G$.  By Lemma~\ref{lemma:lemma4.11}, $m\geq
\|C\|$.  If $C$ is good, then by Lemmas~\ref{lemma:lemma4.11}
and~\ref{lemma:lemma4.12} , we get $m=\|C\|$ in phase~$1$.  If $C$ is
flat, then we get $\|C\|$ in phase~$2$ by Lemma~\ref{lemma:lemma4.13}.
\end{proof}
\paragraph{Proving Lemma~\ref{lemma:lemma4.11}}
\begin{proof}[Proof of Lemma~\ref{lemma:lemma4.11}]
Let $a\in S_1=S_r(a_1,a_2,a_3,b_2,b_3)$ and $b\in
S_2=S_r(b_1,b_2,b_3,a_2,a_3)$.  With
\[
(v_0,\ldots,v_7)=(a_2,b,a_3,b_3,a,b_2,a_1,b_1),
\]
let each $P_i$ with $i\in\{0,\ldots,7\}$ be a shortest
$v_iv_{i^+}$-path in $G$ with $i^+=i+1\mod 8$.  Let
$d(a_1,a_2)=2\ell_1+r_1+r_2$ and $d(b_1,b_2)=2\ell_2+r_1+r_2$.  Since
$T$ is $r$-valid, $\ell_1=\ell_2=\ell$.

(1) $Q_i=G[P_i\cup P_{i^+}]$ is a path for each $i\in\{0,\ldots,7\}$.

By $a\in S_1$ (respectively, $b\in S_2$), we know that $Q_2$, $Q_3$,
and $Q_4$ (respectively, $Q_0$, $Q_1$ and $Q_7$) are paths. Since $T$
is $r$-valid, $Q_5$ and $Q_6$ are paths.

(2) $P_i$ and $P_j$ are anticomplete for each $\{i,j\}\subseteq
\{0,\ldots,7\}$ with $|i-j|>1$.

Since $T$ is $r$-valid, $P_i$ and $P_j$ are anticomplete for each $(i,j)\in\{(2,5),(2,6),(2,7),(5,7)\}$.

Let $a=a_0$ and $b=b_0$ for convenience.  Suppose that $P_x$ and $P_y$
are not anticomplete for an $x\in\{3,4\}$ and $y\in\{0,\ldots,7\}$
with $|x-y|>1$.  Let $P_x=P(a,b_i)$ and $P_y=P(a_j,b_k)$ with
$i\in\{2,3\}$, $j\in\{1,2,3\}$ and $k\in\{0,1,2,3\}\setminus\{i\}$.
Let $\|P_x\|=\ell+r'_1$ and $\|P_y\|=\ell+r'_2$.  Note that $r'_1\leq
r_2$.  We have (i) $d_G(a,a_j) + d_G(b_i,b_k)\leq \|P_x\|+\|P_y\|+2 =
2\ell+2+r'_1+r'_2 \leq 2\ell+2+r_2+r'_2$.  By $a\in S_1$, and
$j\in\{1,2,3\}$ we have (ii) $d_G(a,a_j)\geq 2\ell +
\min(r_1,r_1+r_3,r_2+r_3) \geq 2\ell + r_2$.  By $\{b_i,b_k\}\subseteq
S_2$ and $(i,k)\in\{(2,0),(2,1),(2,3),(3,0),(3,1),(3,2)\}$ , we have
(iii) $d_G(b_i,b_k)\geq 2\ell + \min(r_1+r_3, r_2, r_2+r_3) \geq
2\ell+r_2$.  Combining equations (i),(ii), and (iii), we get
\[
4\ell+2r_2 \leq d_G(a,a_j) + d_G(b_i,b_k)\leq 2\ell+2+r_2+r'_2,
\]
implying $4\leq 2\ell+r_2\leq 2+r'_2\leq 3$, a contradiction.  This
shows that $P_x$ and $P_y$ are anticomplete for each $x\in\{3,4\}$ and
$y\in\{0,\ldots,7\}$ with $|x-y|>1$.  By symmetry between $a$ and $b$,
$P_x$ and $P_y$ are anticomplete for each $x\in\{0,1\}$ and
$y\in\{0,\ldots,7\}$ with $|x-y|>1$.  This completes (2).

By (1) and (2), $G[P_0\cup\cdots\cup P_7]$ is a hole of length
$8\ell+2(r_1+r_2+r_3)$ as required.
\end{proof}
\paragraph{Proving Lemma~\ref{lemma:lemma4.12}}

\begin{proof}[Proof of Lemma~\ref{lemma:lemma4.12}]
Let $C$ be a good shortest even hole with $\|C\|=8\ell+2(r_1+r_2+r_3)$
for a nice $r=(r_1,r_2,r_3)$.  Let
$(a_2,b,a_3,b_3,a,b_2,a_1,b_1)=(v_0,\ldots,v_7)$ be vertices on $C$ in
clockwise order with $d_C(v_i,v_{i^+})=\ell+q_i$, where
$(q_0,\ldots,q_7)=(0,r_2,r_3,r_2,0,r_1,r_3,r_1)$.  Let $a_4=b_2$,
$a_5=b_3$, $b_4=a_2$ and $b_5=a_3$.

(1) Let $u$ and $v$ be vertices of $C$. If $d_C(u,v)>2\ell+r_1+r_3$,
then $d_G(u,v)\geq 2\ell+r_1+r_3$.

Assume a $uv$-path $P$ with $\|P\|\leq 2\ell+r_1+r_3-1$ for
contradiction, implying $\|P\|<\|C\|/4$ and $2\leq\|P\|\leq d_C(u,v)$.
Let $C_1$ and $C_2$ be the $uv$-paths of $C$ with $\|C_1\|\leq
\|C_2\|$.  Since $P$ is not $C$-bad, $G[P\cup C_2]$ is and even hole
with length $\|P\|+\|C_2\|<\|C_1\|+\|C_2\|=\|C\|$, contradiction.

(2) Let $u$ and $v$ be vertices of $C$. If $d_C(u,v)\leq
2\ell+r_1+r_3$, then $d_G(u,v)=d_C(u,v)$.

Assume a $uv$-path $P$ with $\|P\|\leq d_C(u,v)-1$ for contradiction,
implying $\|P\|<\|C\|/4$ and $2\leq\|P\|\leq d_C(u,v)$.  Let $C_1$ and
$C_2$ be the $uv$-paths of $C$ with $\|C_1\|\leq \|C_2\|$.  Since $P$
is not $C$-bad, $G[P\cup C_2]$ is and even hole with length
$\|P\|+\|C_2\|<\|C_1\|+\|C_2\|=\|C\|$, contradiction.

By (1), $d_G(a,a_2),d_G(a_1,a_3)\geq 2\ell+r_1+r_3$.  By (2), all the
equations regarding $a$ are satisfied and $a\in
S_r(a_1,a_2,a_3,a_4,a_5)$.  By symmetry, $b\in
S_r(b_1,b_2,b_3,b_4,b_5)$.

It remains to show that $T$ is $r$-valid.  Let $P_1=P(a_2,b_1)$,
$P_2=P(b_1,a_1)$, $P_3=P(a_1,b_2)$ and $P_4=P(a_3,b_3)$.  By (2),
$\|P_1\|=\ell+r_1$, $\|P_2\|=\ell+r_3$, $\|P_3\|=\ell+r_1$ and
$\|P_4\|=\ell+r_3$.  By (2), $G[P_1\cup P_2]$ and $G[P_2\cup P_3]$ are
both shortest paths.  Assume for contradiction that $P_1$ and $P_3$
are adjacent.  We have (i) $d_G(a_1,a_2)+d_G(b_1,b_2)\leq
\|P_1\|+\|P_3\|+2 = 2\ell+2r_1$.  By (2) applied on $(a_1,a_2)$ and
$(b_1,b_2)$ we have (ii) $4\ell+2r_1+2r_3 = d_G(a_1,a_2) +
d_G(b_1,b_2)$.  Since (i) and (ii) are contradictory, $P_1$ and $P_3$
are anticomplete.  Assume for contradiction that $P_4$ and $P_1$ are
adjacent.  We have (iii) $d_G(b_1,b_3)+d_G(a_2,a_3)\leq
\|P_4\|+\|P_1\|+2 = 2\ell+r_1+r_3+2$.  By (1) and (2) we have (iv)
$2\ell+r_1+r_3 + 2\ell+r_2\leq d_G(b_1,b_3) + d_G(a_2,a_3)$.  (iii)
and (iv) together imply $2\ell+r_2\leq 2$, a contradiction.  Hence
$P_4$ and $P_1$ are anticomplete.  By symmetry $P_4$ and $P_3$ are
anticomplete.  Assume for contradiction that $P_4$ and $P_2$ are
adjacent.  We have (v) $d_G(a_1,a_3)+d_G(b_1,b_3)\leq
\|P_4\|+\|P_2\|+2 = 2\ell+2r_3+2$.  By (1) and (2) we have (vi)
$4\ell+2r_1+2r_3\leq d_G(b_1,b_3) + d_G(a_2,a_3)$.  (v) and (vi) are
contradictory. Hence $P_4$ and $P_3$ are anticomplete.  Therefore, $T$
is indeed $r$-valid.
\end{proof}

\section{Concluding remarks}
\label{section:section5}

Algorithms for induced subgraphs are important and challenging.  We
give improved algorithms for (1) recognizing perfect graphs via
detecting odd holes (2) finding a shortest odd hole, and (3) finding a
shortest even hole.  It is of interest to further reduce the required
time of these three problems.  In particular, the $O(n^{14})$ gap
between Lai et al.'s $O(n^9)$ time for detecting even
holes~\cite[\S6]{LaiLT20} and our $O(n^{23})$ time for finding a
shortest even hole is still quite large.  There are several obstacles
in augmenting Lai et al.'s algorithm, which is basically a faster
implementation of Chang and Lu's $O(n^{11})$-time
algorithm~\cite{ChangL15}, into one for finding a shortest even hole:
(1) An even hole contained by a reported ``beetle'' need not be a
shortest even hole.  (2) Their involved recursive process (based
on~\cite{ConfortiCKV02a,ConfortiCKV02b}) to decompose a graph by
``$2$-joins'' and ``star-cutsets'' may change the length of a shortest
even hole.  (3) It is unclear how to find a shortest even hole in a
$2$-join-free and star-cutset-free graph that is not an ``extended
clique tree''.  Techniques to resolve these issues should be
interesting.

\bibliographystyle{hilabbrv}
\bibliography{hole}
\end{CJK*}
\end{document}